\documentclass{mathincs}
\usepackage[T1]{fontenc}
\usepackage{graphicx}
\usepackage[linesnumbered,ruled,vlined]{algorithm2e}
\DontPrintSemicolon
\SetAlFnt{\footnotesize}
\SetAlCapFnt{\small}
\SetAlCapNameFnt{\small}
\SetArgSty{textnormal}

\usepackage{amssymb}
\usepackage{amsmath}
\usepackage{lineno,hyperref}
\usepackage[natbib=true,style=trad-plain,backend=biber]{biblatex}
\AtEveryBibitem{%
	\ifentrytype{article}{
		\clearname{editor}
		\clearfield{issn}
        \clearfield{pages}
	}{}
	\ifentrytype{inbook}{
		\clearname{editor}
	}{}
	\ifentrytype{incollection}{
		\clearfield{isbn}
        \clearfield{pages}
	}{}
	\ifentrytype{inproceedings}{
		\clearfield{address}
		\clearname{editor}
		\clearfield{eprint}
		\clearfield{isbn}
		\clearlist{location}
		\clearlist{publisher}
        \clearfield{volume}
        \clearfield{pages}
	}{}
	\ifentrytype{proceedings}{
		\clearfield{isbn}
	}{}
	\clearfield{urlyear}
}

\usepackage[nameinlink, capitalise]{cleveref}
\usepackage{fullpage}
\usepackage{color}
\usepackage{enumitem}
\usepackage{subcaption}
\usepackage{todonotes}
\usepackage{wrapfig}
\usepackage{booktabs}

\usepackage{tikz}
\usetikzlibrary{arrows.meta}
\usetikzlibrary{calc,patterns,external}
\usepackage{pgfplots}
\pgfplotsset{compat=newest}
\usepgfplotslibrary{fillbetween}

\newcommand\vdeg[2]{\deg_{#1}(#2)}

\DeclareMathOperator{\levelOp}{level}
\newcommand\level[1]{\levelOp(#1)}

\DeclareMathOperator{\coeffOp}{coeff}
\newcommand\coeff[2]{\coeffOp_{#1}(#2)}

\DeclareMathOperator{\discOp}{disc}
\newcommand\disc[2]{\discOp_{#1}(#2)}
\DeclareMathOperator{\resOp}{res}
\newcommand\res[3]{\resOp_{#1}(#2,#3)}

\DeclareMathOperator{\realRootsOp}{realRoots}
\newcommand\realRoots[1]{\realRootsOp(#1)}

\DeclareMathOperator{\irootOp}{root}
\newcommand\iroot[2]{\irootOp_{#1,#2}}

\newcommand\rationals[0]{\mathbb{Q}}
\newcommand\reals[0]{\mathbb{R}}

\newcommand\posints[0]{\mathbb{N}_{>0}}

\newcommand\sgninv[1]{\textit{sgn\_inv}(#1)}

\newcommand\del[1]{\textit{an\_del}(#1)}

\newcommand\representation[1]{\textit{repr}(#1)}
\newcommand\irordering[1]{\textit{ir\_ord}(#1)}

\newcommand\connected[1]{\textit{connected}(#1)}

\newcommand\sample[1]{\textit{sample}(#1)}
\newcommand\covering[1]{\textit{cov}(#1)}

\newcommand\proj[2]{\ensuremath{#1{\downarrow}_{#2}}}

\newcommand\Isymb[0]{\textnormal{\texttt{I}}}

\newcommand\cellsymb[0]{C}

\newcommand{\I}{\mathbb{I}}
\newcommand{\Q}{\mathbb{Q}}
\newcommand{\R}{\mathbb{R}}

\newcommand{\fals}{\texttt{False}\xspace}
\newcommand{\tru}{\texttt{True}\xspace}
\newcommand{\udef}{\texttt{Undef}\xspace}
\newcommand{\sat}{\texttt{SAT}\xspace}
\newcommand{\unsat}{\texttt{UNSAT}\xspace}
\newcommand{\Matrix}{\ensuremath{\overline{\varphi}}\xspace}
\newcommand{\Prefix}[1]{\ensuremath{Q_{#1} x_{#1} \cdots Q_n x_n}\xspace}

\newtheorem{theorem}{Theorem}[section]
\newtheorem{lemma}[theorem]{Lemma}
\theoremstyle{definition}
\newtheorem{definition}[theorem]{Definition}
\newtheorem*{example}{Example}
\numberwithin{equation}{section}

\SetKwInOut{Input}{Input}
\SetKwInOut{Output}{Output}
\SetKwInOut{Data}{Data}
\SetKwFunction{FRecurse}{recurse}
\SetKwFunction{FExists}{exists}
\SetKwFunction{FForall}{forall}
\SetKwFunction{FParameter}{parameter}
\SetKwFunction{FSampleOutside}{sample\_outside}
\SetKwFunction{FEnclosingInterval}{get\_enclosing\_cell}
\SetKwFunction{FCharacterizeInterval}{characterize\_cell}
\SetKwFunction{FCharacterizeCovering}{characterize\_covering}
\SetKwFunction{FComputeInterval}{compute\_cell}
\SetKwFunction{FImplicant}{implicant\_polynomials}
\SetKwFunction{FComputeCover}{compute\_cover}
\SetKwFunction{FIndexedRootFormula}{indexed\_root\_formula}

\begin{document}

\title[Cylindrical Algebraic Covering for Quantifiers]
 {Extensions of the Cylindrical Algebraic Covering Method for Quantifiers}
\author[Nalbach]{Jasper Nalbach}

\address{%
RWTH Aachen University\\
Aachen\\
Germany}

\email{nalbach@cs.rwth-aachen.de}

\thanks{Jasper Nalbach was supported by the Deutsche Forschungsgemeinschaft (DFG, German Research Foundation) as part of RTG 2236 \emph{UnRAVeL} and AB 461/9-1 \emph{SMT-ART}}
\author[Kremer]{Gereon Kremer}
\address{Certora Ltd\\
Germany}
\email{gereon.kremer@gmail.com}
\subjclass{Primary 14Q99; Secondary 68W99}

\keywords{Non-linear Arithmetic, Cylindrical Algebraic Covering, Quantifier Elimination}

\date{\today}

\begin{abstract}
	The \emph{cylindrical algebraic covering} method was originally proposed to decide the satisfiability of a set of \emph{non-linear real arithmetic} constraints.
	We reformulate and extend the cylindrical algebraic covering method to allow for checking the truth of arbitrary non-linear arithmetic formulas, adding support for both quantifiers and Boolean structure. Furthermore, we also propose a variant to perform \emph{quantifier elimination} on such formulas. After introducing the algorithm, we elaborate on various extensions, optimizations and heuristics. Finally, we present an experimental evaluation of our implementation and provide a comparison with state-of-the-art SMT solvers and quantifier elimination tools.
\end{abstract}

\maketitle

\section{Introduction}

\emph{Non-linear real arithmetic (NRA)} (or \emph{real algebra}) is the first-order theory whose atoms are polynomial constraints over real variables.
We consider three fundamental questions with regard to this theory: (1) \emph{Satisfiability} of quantifier-free formulas; that is, deciding whether an assignment to the formula's variables exists such that the formula evaluates to \tru. (2) \emph{Truth of sentences}; that is, deciding whether formulas where all variables are quantified are equivalent to \tru or \fals. Satisfiability is a special case of this question, as we can existentially quantify all free variables to obtain a sentence. (3) \emph{Quantifier elimination} in formulas containing both free variables (\emph{parameters}) and quantified variables; that is, computing an equivalent quantifier-free formula over the parameters. Deciding the truth of sentences is a special case of this question, as we can eliminate all quantifiers.

The field of \emph{satisfiability-modulo-theories (SMT)} solving deals with the first two problems; while checking the satisfiability of quantifier-free non-linear real arithmetic formula has fairly good support, checking the truth of sentences still lacks accessible and efficient tools. The tools for \emph{quantifier elimination} work a bit differently from SMT tools and might benefit from integrating ideas from SMT solving. 

Tarski~\cite{tarski1951} established the existence of  quantifier elimination methods for non-linear real arithmetic, although his method was practically unusable due its non-elementary complexity bounds. Today, the \emph{cylindrical algebraic decomposition (CAD)}~\cite{collins1975} method is the only complete procedure for answering all these questions that is used in practice, despite its doubly exponential worst-case complexity that severely limits the scalability of the method.
For the satisfiability problem of conjunctions of constraints, motivated by the application in SMT solving, the \emph{cylindrical algebraic covering (CAlC)} method~\cite{abraham2021} has been developed based on cylindrical algebraic decomposition. Although it retains the doubly exponential complexity, its performance is significantly better in practice~\cite{abraham2021,kremer2021calcimpl} while its implementation requires only a simple bookkeeping data structure. Furthermore, it more closely resembles human reasoning and is more amenable to proof production~\cite{abraham2021a,abraham2020}.

\paragraph{Contribution}

This paper extends~\cite{kremer2023calc} in which we propose a novel reformulation and extension of the cylindrical algebraic covering method that goes beyond the satisfiability problem of conjunctions to allow solving arbitrary quantified formulas as well as quantifier elimination queries. This work elaborates the details of the algorithm, proposes some optimizations and extensions, and provides an evaluation of its implementation.
We first consider checking truth where all variables are explicitly quantified, either existentially or universally, in \Cref{sec:quantified}, and then expand to the \emph{quantifier elimination} problem in \Cref{sec:qe}. The presented method does not rely on a SAT solver to solve a Boolean abstraction; we will elaborate in \Cref{sec:implicants} how the Boolean structure is incorporated. Afterwards, we present a divide-and-conquer adaption in \Cref{sec:splitting} and a fine-grained proof system for CAlC in \Cref{sec:proofs}. Finally, we elaborate on heuristics for directing the search in \Cref{sec:heuristics} and evaluate them in \Cref{sec:experiments}, comparing with state-of-the-art tools. We conclude in \Cref{sec:conclusion}.

\section{Related Work}
\label{sec:priorwork}

\subsection{Quantifier-free Formulas}

The \emph{NLSAT}~\cite{jovanovic2012} algorithm, later generalized to the \emph{model-constructing satisfiability calculus (MCSAT)}~\cite{demoura2013mcsat} and implemented in \texttt{z3} and \texttt{yices2}, as well as in our solver \texttt{SMT-RAT}, can be seen as an extension of CDCL(T)~\cite{ganzinger2004} for first-order theories; it is thus a combination of a search for a satisfying solution (exploration) and generalization of unsatisfiable solutions (deduction). In addition to the Boolean reasoning, we make decisions and propagations on theory variables as well, i.e., we assign values to theory variables and then evaluate and propagate literals based on these assignments. The assignment of theory variables is required to not violate any theory constraint that is assumed to be true on the Boolean level. If such an assignment is not possible for a theory variable $x_{n+1}$, then the unsatisfying intervals of the theory constraints (partially evaluated using the values for $x_1,\ldots,x_n$) cover the whole real line $\R$. This conflict is generalized using CAD-based techniques to a formula that excludes a connected set $S$ in $\R^n$ such that the cylinder $S \times \R$ is covered by unsatisfying sets of the same constraints.

The \emph{cylindrical algebraic covering (CAlC)} method~\cite{abraham2021} is implemented both in \texttt{cvc5} and in \texttt{SMT-RAT}.
It has has a number of similarities to the NLSAT algorithm, but is a theory solver for the classical CDCL(T) architecture and thus only considers conjunctions of constraints instead of arbitrary formulas.
The bookkeeping required is simpler than for MCSAT: if a conflict occurs, it is not generalized to a formula in $x_1,\ldots,x_n$, but uses an implicit description of the unsatisfying values, which consists of a set of CAD projection polynomials and the sample $s \in \R^n$ that caused the conflict; this implicit description is used to deduce unsatisfying intervals, similarly to how they are gathered from the input constraints.
While MCSAT allows for both breadth-first and depth-first search (or a combination of both), CAlC is restricted to a depth-first search due to its simpler bookkeeping.

\subsection{Quantified Formulas}

To decide sentences including quantifiers, the \texttt{z3} and \texttt{cvc5} solvers implement approaches based on \emph{quantifier instantiation}~\cite{demoura2007,niemetz2021,bjorner2015,ge2009} which are \emph{incomplete} for non-linear real arithmetic. We do not go into detail here, as these approaches work differently from our algorithm; for an overview, we refer to the related work section in~\cite{niemetz2021}.

A recent \emph{complete} algorithm is the \emph{QSMA} algorithm~\cite{bonacina2023} implemented in \texttt{yicesQS}. It maintains for every quantified subformula $Q y. \varphi(x_1,\ldots,x_n,y)$ (where $Q \in \{ \exists, \forall \}$) an over- and under-approximation (encoded as formulas in the free variables $x_1,\ldots,x_n,y$ of $\varphi$) which are refined iteratively in a breadth-first search manner. This refinement continues until these approximations are fine-grained enough to conclude satisfiability or unsatisfiability. The required approximations can be generated by NLSAT-based solvers.

For quantifier elimination, mentionable tools implementing CAD that are readily available include \texttt{QEPCAD~B}~\cite{brown1999,brown1999a,brown2003} and \texttt{Redlog}~\cite{seidl2003}. In particular, the former applies various techniques to reduce the number of projection polynomials in the CAD used to construct the solution formulas, which is often coarser than the result of naive CAD algorithms.
Furthermore, the commercial tools \texttt{Maple}~\cite{iwane2009} and \texttt{Mathematica}~\cite{strzebonski2000} implement the CAD for quantifier elimination. We refer to the respective sources for more details.

\section{Preliminaries}

Let $s \in \reals^i$, $s' \in \reals$, and $I \subseteq \reals$. We denote the Cartesian product of two sets using the $\times$ symbol. We denote by $s \times s'$ the point $(s_1,\ldots,s_i,s') \in \reals^{i+1}$, by $s_{[j]}$, $j<i$ the point $(s_1,\ldots,s_j) \in \reals^j$, and by $s \times I$ the Cartesian product $\{ s \} \times I \subseteq \reals^{i+1}$. Let $R \subseteq \reals^i$, then $\proj{R}{[j]}$, $j<i$ denotes the projection of $R$ onto the first $j$ coordinates. A set $R \subseteq \reals^i$ is called a \emph{cell} if it is connected.
Given a tuple such as $t = (a,b,c)$, we use $t.a$, $t.b$ and $t.c$ to access its entries.

Let $x_1, \dots, x_n$ be \emph{variables}, then $\rationals[x_1,\ldots,x_i]$ denotes the set of all \emph{polynomials} with rational coefficients in $x_1,\ldots,x_i$. We assume an ordering on the variables $x_1 \prec \cdots \prec x_n$. The highest variable (w.r.t. this ordering) occurring in a polynomial $p$ is called the \emph{main variable} and its \emph{level} $\level{p}$ is the index of the main variable.
The \emph{degree} of $p$ in $x_i$ is denoted by $\vdeg{x_i}{p}$. Let $s \in \R^j$, $j \leq i$, then $p(s,x_{j+1},\ldots,x_i)$ denotes the polynomial after substituting $s$ into $p$. The polynomial $p$ is called \emph{nullified} over $s$ if $j<i$ and $p(s,x_{j+1},\ldots,x_i) = 0$. Let $p \in \rationals[y]$ be a univariate polynomial, then $\realRoots{p}$ denotes the set of \emph{real roots} of $p$ in $y$.

A polynomial constraint $p \sim 0$ compares a polynomial $p$ with zero using a relation symbol $\sim \in \{=,\leq,\geq,\neq,<,>\}$; notations for polynomials are transferred to constraints where meaningful. We assume every formula $\varphi$ to be a first-order formula over non-linear real arithmetic with polynomial constraints defined in variables $x_1, \dots, x_n$. We use $\equiv$ do denote equivalence of formulas (modulo non-linear real arithmetic). A cell $R \subseteq \reals^i$ is called \emph{semi-algebraic} if it is the solution set of a quantifier-free non-linear real arithmetic formula. A polynomial $p$ is called \emph{sign-invariant} on $R$ if for all points in $R$ the evaluation of $p$ has the same sign. A formula $\varphi$ is called \emph{truth-invariant} on $R$ if for all points in $R$ the evaluation of $\varphi$ is equivalent to the same truth value.

Let $s \in \R^i$, then we denote the \emph{(partial) evaluation up to level $i$} of $\varphi$ over $s$ by $\varphi[s]$: constraints of level $i$ evaluate to $\tru$ or $\fals$ according to standard semantics, otherwise they evaluate to $\udef$ (i.e., under this partial evaluation $x_1\cdot x_2 > 1$ evaluates to $\udef$ at $x_1 = 0$). The semantics are extended for formulas inductively according to the three-state semantics of the logical operators; that means, we evaluate formulas also on sample points which do not assign all variables in the formula, e.g. ``$\udef \wedge \fals$'' evaluates to $\fals$.

A formula $\varphi$ is in \emph{prenex normal form} if it consists of a \emph{prefix} of quantifiers and a quantifier-free formula called the \emph{matrix} \Matrix:
\[
	\varphi := \Prefix{k+1}.\; \Matrix(x_1, \dots, x_n)
\]
If $k \neq 0$, $\varphi$ has free variables (also called \emph{parameters}) that are not explicitly quantified.

\subsection{Cylindrical Algebraic Decomposition}

We give a short introduction to CAD. For the understanding of some details of this paper, we assume the reader to be familiar with CAD (we refer for instance to~\cite{kremer2020phd}).
We first formally define what the output of a CAD computation is:

\begin{definition}[Cylindrical Algebraic Decomposition]
	\label{def:cad}

	Let $i \in \posints,\ i \leq n$.

	Let $R \subseteq \reals^i$ be a cell.
	$R$ is called \emph{locally cylindrical} if either $i=1$; or $i>1$, $\proj{R}{[i-1]}$ is locally cylindrical, and there exist continuous functions $\theta_1,\theta_2: \proj{R}{[i-1]} \rightarrow \reals$ such that either $R = \{ (r,\theta_1(r)) \mid r \in \proj{R}{[i-1]} \}$ ($R$ is called a \emph{sector} over $\proj{R}{[i-1]}$) or $R = \{ (r,r') \mid r \in \proj{R}{[i-1]}, \theta_1(r) < r' < \theta_2(r) \}$ ($R$ is called a \emph{section} over $\proj{R}{[i-1]}$).

	Let $D \subset \{ R \mid R \subseteq \reals^i \}$ be finite.
	$D$ is called a \emph{decomposition} if $\cup_{R \in D} R = \reals^i$ and $R \cap R' = \emptyset$ for all $R,R' \in D,\ R \neq R'$.
	$D$ is called a \emph{cylindrical decomposition} if $D$ is a decomposition, each $R \in D$ is locally cylindrical; and either $i=1$, or $i>1$ and $D' := \{ \proj{R}{[i-1]} \mid R \in D \}$ is a cylindrical decomposition.
	$D$ is called a \emph{cylindrical algebraic decomposition (CAD)} if it is a cylindrical decomposition and each $R \in D$ is semi-algebraic.
\end{definition}

A CAD $D$ is computed such that it is \emph{sign-invariant for $P \subset \rationals[x_1,\ldots,x_n]$}, i.e. every $p \in P$ is sign-invariant in every $R \in D$. A CAD is useful because of its cylindrical structure: Firstly, local cylindricity of each cell allows to iteratively compute a point that is contained in the cell. Secondly, the cylindrical arrangement of the cells (that is, given a cell $R' \in D'$, we know that there are $R_1,\ldots,R_k \in D$ such that the cylinder $R' \times \reals$ over $R'$ is equal to $R_1 \cup \ldots \cup R_k$) allows us to systematically explore the decomposition and reason about its structure.

A sign-invariant CAD is computed recursively: Given a set $P \subset \rationals[x_1,\ldots,x_i]$ of polynomials, we compute a projection $P' \subset \rationals[x_1,\ldots,x_{i-1}]$ such that, given a sign-invariant CAD $D'$ of $P'$, we can compute a CAD $D$ for $P$. To guarantee the latter, we define the projection such that the set $P$ is \emph{delineable} on each cell $R' \in D'$, that is, the roots $\{ (r',r) \mid r' \in R', r \in \realRoots{p(r',x_i)}, p \in P \}$ of $P$ on $R'$ can be described by continuous \emph{root functions} $\theta_1,\ldots,\theta_k: R' \to \reals$ such that $\theta_1(r') < \ldots < \theta_k(r')$ for all $r' \in R'$, i.e. these roots do not intersect. These functions then define the bounds of the cells in $D$ that decompose the cylinder $R' \times \reals$.

A \emph{projection operator} maps polynomials of level $i$ to polynomials of lower level. Ingredients are \emph{discriminants} $\disc{x_i}{p}$ which -- together with some \emph{coefficients} $\coeff{x_i}{p}$ -- guarantee delineability of a single polynomial $p \in \rationals[x_1,\ldots,x_i]$ (i.e. the root functions of $p$ are well-defined on $R'$ and do not intersect), and \emph{resultants} $\res{x_i}{p}{q}$ which guarantee that the root functions of two polynomials $p,q \in \rationals[x_1,\ldots,x_i]$ are either equal or do not intersect on $R'$.

The CAD algorithm has doubly exponential complexity~\cite{brown2007} due to the computation of iterated resultants (and discriminants) which results in a doubly exponential growth of degrees~\cite{england2016}. While this growth is inherent to the CAD, different projection operators with vastly different output sizes exist. If a projection operator computes more polynomials than necessary, the subsequent projection steps and computations oftentimes amplify this overhead drastically and thereby impact the efficiency in practice.

The original operator by Collins~\cite{collins1975} is complete but expensive. McCallum's projection operator~\cite{mccallum1998} is more efficient but incomplete (and was later improved by Brown~\cite{brown2001}), i.e., it fails if any polynomial in the projection is nullified over some point. It also does not compute a sign-invariant CAD, but an \emph{order-invariant} CAD, which is a property of polynomials which is stronger than sign-invariance (as the formal details are not required for this paper, we refer to~\cite{mccallum1998} for them). Lately, Lazard's projection operator~\cite{lazard1994}, which is complete and similarly efficient as McCallum, has been proven to be correct~\cite{mccallum2019}. Although this work is applicable to all mentioned projection operators, the projection operator defined in this paper is based on McCallum's projection due to its efficiency and simplicity; the presented algorithm detects nullifications and returns an unknown result.

\subsection{Cylindrical Algebraic Covering}

A CAD already allows to answer our questions about a non-linear real arithmetic formula $\varphi$: To do so, we compute a sign-invariant CAD $D$ for all polynomials occurring in $\varphi$; $D$ will be truth-invariant for $\varphi$. However, the CAD is finer than we actually need. We thus define the notion of \emph{cylindrical algebraic covering (CAlC)}, which is defined analogously to \Cref{def:cad}, but does not require that its cells are disjoint. We now compute a covering $C$ that is truth-invariant for $\varphi$, but not necessarily sign-invariant for its defining polynomials!

When computing a CAD, we fully decompose the space top-down, i.e., compute a full projection, and then construct the resulting cells. The CAlC is computed bottom-up: We compute sample points and generalize the reasons for the truth-value of $\varphi$ at these points. 

We briefly present the idea behind the CAlC method for checking the \emph{existential fragment} of non-linear real arithmetic (i.e. the input is in prenex normal form and all variables are existentially quantified) and refer to~\cite{abraham2021} for more details and to the preliminaries of~\cite{bar2023} for another intuitive introduction.

The fundamental idea is to recursively construct a (partial) sample point and collect intervals that represent unsatisfiable cells above this sample point.
When a sample point can not be extended because these intervals form a covering of the real line in the next dimension, the covering is projected into the previous dimension to refute the current sample point.
We then backtrack and choose a different value for the variable on the highest level.
Eventually, either a full sample point is constructed, and we return \sat; or an unsatisfiable covering is constructed in the first dimension, and we return \unsat.
In contrast to cells from cylindrical algebraic decomposition, intervals do not form a decomposition as they may overlap.

The algorithm starts by constructing unsatisfiable intervals for $x_1$ based on univariate constraints and then tries to select a value $s_1$ for the variable $x_1$ outside these intervals. If such a value exists, the method is called recursively with the partial sample point $(s_1)$.
After substituting $x_1 = s_1$, the constraints with main variable $x_2$ become univariate and thus suitable for identifying unsatisfiable intervals for $x_2$.
This process is continued recursively until either all constraints are satisfied (and we return \sat) or for some $x_i$ no suitable value exists. In the latter case, the set of unsatisfiable intervals covers the whole real line and forms a covering. This covering is generalized by projecting it to dimension $i-1$. The idea is to use projection tools borrowed from cylindrical algebraic decomposition with some improvements: as we only need to characterize this covering and not a decomposition, only a subset of the full projection is needed.
Using the current sample point, an interval for the variable $x_{i-1}$ with respect to the projection result can be computed which is added to the set of unsatisfiable intervals for $x_{i-1}$, possibly taking part in an unsatisfiable covering for $x_{i-1}$.
We now try another value for $x_{i-1}$, respecting the current set if unsatisfiable intervals. Unless we find a full satisfying sample point we eventually obtain an unsatisfiable covering for the first variable $x_1$ and return \unsat.

\paragraph{Implicit Cells.}

We generalize intervals (over a partial sample point) by attaching algebraic information in the form of sets of polynomials whose order-invariance characterizes satisfiability-invariant cells of a multivariate formula.

\begin{definition}[Implicit Cell]
	Let $i \in \posints$, $P \subseteq \rationals[x_1,\ldots,x_i]$ be a set of polynomials, $s \in \reals^i$, and $I \subseteq \reals$ be an interval. Let $R \subseteq \reals^i$ be the maximal connected subset containing $s$ where all polynomials in $P$ are order-invariant.
	An \emph{implicit cell of level $i$} is a tuple $(P, s, I)$ such that $I = \{ r \mid (s_1,\ldots,s_{i-1},r) \in R \}$, i.e., its $i$-th coordinate is bounded from below by is the greatest root of $P$ in $x_i$ below (or equal to) $s_i$ and bounded from above by the smallest root above (or equal to) $s_i$.
	$R$ is called the \emph{cell defined by $P$ and $s$}.
\end{definition}

\begin{example}
	Consider the polynomials $P = \{ x_2+1,x_1^2+x_2^2-2,x_1-1\}$ and the sample point $s=(0,0)$. Then the corresponding implicit cell is $(P,s,(-1,\sqrt{2}))$, as $x_2+1$ has a zero at $(0,-1)$, $x_1^2+x_2^2-2$ at $(0,\sqrt{2})$, and no polynomial has a zero over $(s_1)=(0)$ between these two points.
\end{example}

\paragraph{Implicants.}

For reasoning about the Boolean structure of a formula, we introduce the notion of implicants.
An \emph{implicant} $\psi$ of a quantifier-free formula $\varphi$ is usually understood to be a ``simpler'' quantifier-free formula that implies $\varphi$ ($\psi \Rightarrow \varphi$), and the set of constraints of $\psi$ is a subset of the ones in $\varphi$.
We adapt this concept as follows.

\begin{definition}[Implicant]
	\label{def:implicant}
	Let $s \in \R^i$ be a (partial) sample point and $\varphi$ be a quantifier-free formula in $n$ variables. If $\varphi[s] \equiv \tru$, then $\psi$ is an \emph{implicant of $\varphi$ with respect to $s$} if
	\[ \psi[s] = \tru \land \left( \psi \Rightarrow \varphi \right) \]
	and the constraints of $\psi$ are of level at most $i$ and contained in $\varphi$.
	If $\varphi[s] \equiv \fals$, then $\psi$ is an \emph{implicant of $\varphi$ with respect to $s$} if
	\[ \psi[s] = \tru \land \left( \psi \Rightarrow \neg\varphi \right) \]
	and the constraints of $\psi$ are of level at most $i$ and contained in $\varphi$.
	We call $\psi$ a \emph{prime implicant} of $\varphi$ if the set of constraints in $\psi$ is minimal among all implicants of $\varphi$.
\end{definition}

Note that in the above definition, we allow $\varphi[s]$ to be a tautology or contradiction, while we require that the implicant $\psi[s]$ evaluates to a truth value after plugging in $s$.

\begin{example}
	Let $\varphi = x^2>0 \wedge (x<2 \vee x>4)$. Note that $\varphi(1) = \tru$, $\varphi(3) = \fals$ and $\varphi(0) = \fals$. $x^2>0 \wedge x < 2$ is a prime implicant of $\varphi$ w.r.t. $1$. $\neg (x<2 \vee x>4)$ is a prime implicant of $\varphi$ w.r.t. $3$. Both $\neg(x^2>0)$ and $\neg(x^2>0 \wedge x>4)$ are implicants of $\varphi$ w.r.t. $0$, but only the first is a prime implicant.

	Let $\varphi = (x<0 \vee y\leq4)\wedge (x>2 \vee y>4)$. Note that $\varphi(1,y) \equiv \fals$. $\neg(x<0) \wedge \neg(x>2)$ is a prime implicant of $\varphi$ w.r.t. $1$.
\end{example}

\section{Quantified Problems}
\label{sec:quantified}

We first describe how the cylindrical algebraic covering method can be adapted for problems where all variables are quantified (again, assuming the input formula is in prenex normal form).
Our presentation follows the structure of~\cite{abraham2021}, but is different in some details.

\begin{algorithm}[ht!]
	\Data{Global prefix \Prefix{1} and matrix \Matrix.}
	\Output{Either \sat or \unsat}

	$(f,O) := $ \FRecurse{$()$} \tcp*{\Cref{alg:main}} 
	\Return{$f$}
	\caption{\texttt{user\_call()} \label{alg:user}}
\end{algorithm}

\begin{algorithm}[ht!]
	\Data{Global prefix \Prefix{1} and matrix \Matrix.}
	\Input{Sample point $s=(s_1,\ldots,s_{i-1}) \in \R^{i-1}$.}
	\Output{$(\sat, \cellsymb)$ or $(\unsat, \cellsymb)$ where $s \times \cellsymb.I$ can or can not be extended to a model for any $s_i \in \cellsymb.I$. In both cases, $\cellsymb$ describes how $s$ can be generalized.
	}

    \lIf{$Q_i = \exists$}{
        \Return{\FExists{$s$}} \tcp*[f]{\Cref{alg:main_exists}}
    }
	\lElse{
        \Return{\FForall{$s$}} \tcp*[f]{\Cref{alg:main_forall}}
    }

	\caption{\texttt{recurse($s$)} \label{alg:main}}
\end{algorithm}

\subsection{Existential Quantification}

\begin{algorithm}[ht!]
	\Data{Global prefix \Prefix{1} and matrix \Matrix.}
	\Input{Sample point $s=(s_1,\ldots,s_{i-1}) \in \R^{i-1}$.}
	\Output{see \Cref{alg:main}}

	$\I_\text{unsat} := \emptyset$ \; \label{alg:main_exists:init}
	\While{$\bigcup_{\cellsymb \in \I_\text{unsat}} \cellsymb.I \neq \R$}{\label{alg:main_exists:while}
		$s_i :=$ \FSampleOutside{$\I_\text{unsat}$}  \;\label{alg:main_exists:sample}
		\uIf{$\Matrix[s \times s_i] = \fals$}{
			$(f,O) := (\unsat,\FEnclosingInterval{$s \times s_i$})$ \tcp*{\Cref{alg:get-enclosing-interval}}\label{alg:main_exists:unsat}
		}
		\uElseIf{$\Matrix[s \times s_i] = \tru$}{
			$(f,O) := (\sat,\FEnclosingInterval{$s \times s_i$})$ \tcp*{\Cref{alg:get-enclosing-interval}}\label{alg:main_exists:sat}
		}
		\Else(it holds $i<n$){
			$(f,O) :=$ \FRecurse{$s \times s_i$} \tcp*{\Cref{alg:main}, recursive call}
		}
		\uIf{$f = \sat$ }{
			$\cellsymb := $ \FCharacterizeInterval{$s$, $O$} \tcp*{\Cref{alg:characterize-interval}}\label{alg:main_exists:characterize}
			\Return{$(\sat, \cellsymb)$}
		} 
		\ElseIf{$f = \unsat$ }
		{
			$\I_\text{unsat} := \I_\text{unsat} \cup \{O\}$ \;\label{alg:main_exists:collect}
		}
	}

    $\cellsymb :=$ \FCharacterizeCovering{$s$, $\I_\text{unsat}$} \tcp*{\Cref{alg:characterize-covering}}\label{alg:main_exists:characterize2}

	\Return{$(\unsat, \cellsymb)$}
	
	\caption{\texttt{exists($s$)} \label{alg:main_exists}}
\end{algorithm}

First assume that all variables are existentially quantified, thus our algorithm resembles the original from~\cite{abraham2021}: \Cref{alg:main_exists} (which corresponds to \texttt{get\_unsat\_cover} from~\cite{abraham2021}) is recursively called to choose a suitable value $s_i$ for the next variable $x_i$ such that the resulting sample point does not conflict with the formula $\Matrix$. The algorithm maintains a list of cells $\I_\text{unsat}$ in $\reals^i$ that are known to violate the formula. The call to \texttt{sample\_outside} in \Cref{alg:main_forall:sample} chooses a value outside of $\I_\text{unsat}$.

If $s \times s_i$ immediately evaluates $\Matrix$ to $\fals$, we use \Cref{alg:get-enclosing-interval} that generalizes this unsatisfying sample to an unsatisfying cell in $\reals^i$ (see below for details); in~\cite{abraham2021}, such cells would be computed by \texttt{get\_unsat\_intervals}. If $s \times s_i$ immediately evaluates $\Matrix$ to $\tru$, the method generalizes the satisfying sample to a satisfying cell in $\reals^i$; in~\cite{abraham2021} we would return the satisfying sample only. If the formula does not evaluate to a truth value, we pick a value for the next variable $x_{i+1}$ by a recursive call. If the recursive call returns a satisfying cell $R \subseteq \reals^i$, we compute its (CAD-style) projection $R' \subseteq \reals^{i-1}$ suitable for the caller (which searches for a value for $x_{i-1}$) using \Cref{alg:characterize-interval}. If the recursive call returns an unsatisfying cell, we add it to the list $\I_\text{unsat}$. The differences to the original algorithm are due to the support for quantifier alternations, for which we need to generalize both satisfying and unsatisfying sample points.

If the list of unsatisfying cells covers the whole real line $s \times \reals$ above the given sample point $s \in \reals^{i-1}$, we compute the projection $R' \subseteq \reals^{i-1}$ (being an unsatisfiable cell on the lower levels) of a cylinder $R' \times \reals$ covered by unsatisfiable cells in $\reals^i$ represented by $\I_\text{unsat}$ using \Cref{alg:characterize-covering}.

We note that if $i=1$, \Cref{alg:characterize-interval,alg:characterize-covering} would need to return a cell on the ``zero-th level''. To simplify the presentation, we assume that a special placeholder value is returned instead of an actual interval.

We emphasize that compared to \texttt{get\_unsat\_cover} from~\cite{abraham2021}, there are no significant algorithmical differences; we mere changed some details such that the algorithm can be extended more easily for quantifier alternation.

\subsection{Universal Quantification}

\begin{algorithm}[t]
	\Data{Global prefix \Prefix{1} and matrix \Matrix.}
	\Input{Sample point $s=(s_1,\ldots,s_{i-1}) \in \R^{i-1}$.}
	\Output{see \Cref{alg:main}}

	$\I_\text{sat} := \emptyset$ \;\label{alg:main_forall:init}

	\While{$\bigcup_{\cellsymb \in \I_\text{sat}} \cellsymb.I \neq \R$}{\label{alg:main_forall:while}
		$s_i :=$ \FSampleOutside{$\I_\text{sat}$}  \; \label{alg:main_forall:sample}
		\uIf{$\Matrix[s \times s_i] = \fals$}{
			$(f,O) := (\unsat,\FEnclosingInterval{$s \times s_i$})$ \tcp*{\Cref{alg:get-enclosing-interval}}
		}
		\uElseIf{$\Matrix[s \times s_i] = \tru$}{
			$(f,O) := (\sat,\FEnclosingInterval{$s \times s_i$})$ \tcp*{\Cref{alg:get-enclosing-interval}}
		}
		\Else(it holds $i<n$){
			$(f,O) :=$ \FRecurse{$s \times s_i$} \tcp*{\Cref{alg:main}, recursive call}\label{alg:main_forall:call}
		}
		\uIf{$f = \sat$}{
            $\I_\text{sat} := \I_\text{sat} \cup \{O\}$ \;\label{alg:main_forall:collect}
		} 
		\ElseIf{$f = \unsat$}
		{
            $\cellsymb :=$ \FCharacterizeInterval{$s$, $O$} \tcp*{\Cref{alg:characterize-interval}}
            \Return{$(\unsat, \cellsymb)$} \;\label{alg:main_forall:returnunsat}
		}
	}
	$\cellsymb := $ \FCharacterizeCovering{$s$, $\I_\text{sat}$} \tcp*{\Cref{alg:characterize-covering}}
    \Return{$(\sat, \cellsymb)$}

	\caption{\texttt{forall($s$)} \label{alg:main_forall}}
\end{algorithm}

Now assume the formula contains quantifier alternations. \Cref{alg:user} is the interface to the recursive \Cref{alg:main}, calling it with an empty sample point and extracting the main return value. \Cref{alg:main} checks the current quantifier and calls out to \Cref{alg:main_exists} or \Cref{alg:main_forall} accordingly.

\Cref{alg:main_forall} is mostly identical to \Cref{alg:main_exists}: While \Cref{alg:main_exists} collects unsatisfiable cells and returns early when it finds a satisfiable cell, \Cref{alg:main_forall} collects satisfiable cells and returns early when it finds an unsatisfiable cell.
Note that we project cells and coverings of cylinders (i.e. calling \Cref{alg:characterize-interval} and \Cref{alg:characterize-covering}) for both satisfiable and unsatisfiable coverings in the very same way. That is, the projection operations work with implicit cells; whether the input formula is \tru or \fals on an implicit cell is irrelevant for the projection operations.

\subsection{Truth-Invariant Cells}

\begin{algorithm}[t]
	\Data{Global matrix \Matrix.}
	\Input{Sample point $s \in \R^{i}$ and polynomials $P \subseteq \rationals[x_1,\ldots,x_i]$.}
	\Output{A maximal interval $I \subseteq \reals$ such that $P$ is sign-invariant in $s \times I$.}  

	\lIf{$\exists p \in P.\; (\level{p}=i \wedge p(s_{[i-1]},x_i) = 0)$}{
		\textbf{fail}\tcp*[f]{We fail on nullification}
	}
	$Z := \{ -\infty, \infty \} \cup \bigcup_{p \in P, \level{p}=i} \realRoots{p(s_{[i-1]},x_i)}$ \;
	\lIf{$s_i \in Z$}{ \Return{$[s_i,s_i]$} }
	$l := \max \{ z \in Z \mid z \leq s_i \}$ \;
	$u := \min \{ z \in Z \mid z \geq s_i \}$ \;
	\Return{$(l,u)$}

	\caption{\texttt{compute\_cell($s$, $P$)}}\label{alg:compute-interval}
\end{algorithm}

\begin{algorithm}[t]
	\Data{Global matrix \Matrix.}
	\Input{Sample point $s \in \R^{i}$ such that $\Matrix[s] \equiv \fals$ or $\Matrix[s] \equiv \tru$.}
	\Output{A satisfiability-invariant implicit cell $\cellsymb$ containing $s$.}  
	$P := \FImplicant{$\Matrix$, $s$}$ \;
	\textbf{replace} $P$ by its irreducible factors \;
	$\cellsymb := (P,s,\FComputeInterval{$s$, $P$})$ \tcp*{\Cref{alg:compute-interval}}
	\Return{\cellsymb} 

	\caption{\texttt{get\_enclosing\_cell($s$)}}\label{alg:get-enclosing-interval}
\end{algorithm}

\Cref{alg:get-enclosing-interval} computes an implicit cell around the given sample point that is satisfiability-invariant with respect to \Matrix. It first obtains the set of polynomials from an implicant of \Matrix w.r.t. the current sample $s$ by calling \texttt{implicant\_polynomials}; the sign-invariance of these polynomials directly implies the truth-invariance of \Matrix.
The algorithm then uses \Cref{alg:compute-interval} (which corresponds to \texttt{interval\_from\_characterization} in~\cite{abraham2021}) to construct the interval above $s$ that is contained in a sign-invariant (or truth-invariant) cell.
The helper function \texttt{implicant\_polynomials} is expected to return the polynomials of a (possibly prime) implicant of $\Matrix$ with respect to $s$.
This might include polynomials not only with main variable $x_i$, but also lower-level polynomial, effectively bounding also lower-level coordinates of the sign-invariant cell.

If $\Matrix[s] = \fals$ and $\Matrix$ is a simple conjunction, it is easy to obtain a \emph{prime implicant} as the negation of a single conflicting constraint in \Matrix; calling it in a loop as done in \cref{alg:main_exists} thus emulates $\texttt{get\_unsat\_intervals}$ from~\cite{abraham2021}. If $\Matrix[s] = \tru$ and $\Matrix$ is a simple conjunction and non-redundant (i.e. no sub-formula of $\Matrix$ implies $\Matrix$), then $\Matrix$ itself is the only prime implicant.

\Cref{alg:characterize-interval,alg:characterize-covering} implement a reduced CAD projection based on McCallum's projection operator (see \Cref{sec:projection} for details); these algorithms define the same projection as~\cite[Algorithm~4]{abraham2021}, but split the projection into the characterization of individual cells and the characterization of a covering of a cylinder. \Cref{alg:characterize-interval} computes a CAD projection $R \subseteq \reals^i$ of a single cell in $\reals^{i+1}$ and uses \Cref{alg:compute-interval} to construct an interval for $x_i$. \Cref{alg:characterize-covering} first calls the auxiliary method $\texttt{compute\_cover}$ which takes a set of cells $\I$ as input, and returns a sequence of a subset of these cells: Firstly, it iteratively eliminates ``redundant'' cells (a cell $\cellsymb$ is redundant in $\I$ if $\cellsymb.I \subseteq \cellsymb'.I$ for some other interval $\cellsymb' \in \I$); secondly, it sorts the cells by their interval's lower bound. This way, it is guaranteed that neighboring cells in the resulting sequence overlap (or their union is connected); this is important for the correctness of the projection in \Cref{proj-overlap2}. For more details, we refer to~\cite[Section 4.4.1]{abraham2021}.

\begin{algorithm}[t]
	\Input{Sample point $s \in \R^i$ and an implicit cell $\cellsymb=(\cdot,s \times \cdot,\cdot)$ of level $i+1$.}
	\Output{A satisfiability-invariant implicit cell $\cellsymb'$ containing $s$.}
	$P_{i+1} := \{ p \in P \mid \level{p} = i+1 \}$, $P_\bot := P \setminus P_{i+1}$ \;
	$P' := P_\bot \cup  \{ \disc{x_{i+1}}{p} \mid p \in P_{i{+}1} \} \cup \bigcup_{p \in P_{i{+}1}} \coeff{x_{i+1}}{p} $ \label{alg:characterize-interval:coeffs}\;
	$P' := P' \cup \{ \res{x_{i+1}}{p}{q} \mid p, q \in P_{i{+}1}, p(s \times \cellsymb.I.l) = 0 , \exists s' \leq \cellsymb.I.l.\; q(s \times s') = 0  \}$ \nllabel{proj-overtakel}\;
	$P' := P' \cup \{ \res{x_{i+1}}{p}{q} \mid p, q \in P_{i{+}1}, p(s \times \cellsymb.I.u) = 0 , \exists s' \geq \cellsymb.I.u.\; q(s \times s') = 0  \}$ \nllabel{proj-overtakeu}\;
	$P' := P' \cup \{ \res{x_{i+1}}{p}{q} \mid p, q \in P_{i{+}1}, p(s \times \cellsymb.I.l) = 0, q(s \times \cellsymb.I.u) = 0  \}$ \nllabel{proj-overtake}\;
	\textbf{replace} $P'$ by its irreducible factors \;
	\Return{$(P',s,\FComputeInterval{$s$, $P'$})$} \tcp*{\Cref{alg:compute-interval}}
	
	\caption{\texttt{characterize\_cell($s$, $\cellsymb$)} \label{alg:characterize-interval}}
\end{algorithm}

\begin{algorithm}[t]
	\Input{Sample point $s \in \R^i$ and a set $\I$ of implicit cells of level $i+1$ such that $\bigcup_{\cellsymb \in \I} \cellsymb.I = \R$ and for all $\cellsymb \in \I$ it holds $\cellsymb=(\cdot,s \times \cdot,\cdot)$.}
	\Output{A satisfiability-invariant implicit cell $\cellsymb'$ of level $i$ containing $s$.}
	$(\cellsymb_1, \ldots, \cellsymb_k) :=$ \FComputeCover{$\I$} \tcp*{\cite[Section 4.4.1]{abraham2021}}
	$P' := \bigcup_{j \in \{1,\ldots,k\}} \FCharacterizeInterval{$s$, $\cellsymb_j$}.P$ \tcp*{\Cref{alg:characterize-interval}}
	\For{$j \in \{1, \dots, k-1\}$\nllabel{proj-overlap1}}{
		$P' := P' \cup \{ \res{x_{i+1}}{p}{q} \mid p \in \cellsymb_j.P, p(s \times \cellsymb_j.I.u) = 0, q \in \cellsymb_{j{+}1}.P, q(s \times \cellsymb_{j{+}1}.I.l) = 0 \}$ \nllabel{proj-overlap2}\;
	}
	\textbf{replace} $P'$ by its irreducible factors \;
	\Return{$(P',s,\FComputeInterval{$s$, $P'$})$} \tcp*{\Cref{alg:compute-interval}}
	
	\caption{\texttt{characterize\_covering($s$, $\I$)} \label{alg:characterize-covering}}
\end{algorithm}

\subsection{Details of the projection operator}
\label{sec:projection}

We changed how we normalize the polynomial sets after projection: While~\cite{abraham2021} assumes ``standard CAD simplifications'' of the polynomial sets, we explicitly use the set of their irreducible factors in \cref{alg:get-enclosing-interval}, \cref{alg:characterize-interval}, and \cref{alg:characterize-covering} to satisfy the requirements of the projection operator. Merely using an irreducible square-free basis, the common standard formulation for CAD projection, is not quite sufficient for cylindrical algebraic covering: we eventually compute resultants of polynomials that come from different local projection sets, i.e. from different bases.
If carefully executed, these sets can be made ``pairwise square-free'', as mentioned in~\cite[Section 2.1]{kremer2022}.
Fully factoring all polynomials is more robust and probably even more efficient in practice, if the implementation at hand has this capability.

\Cref{alg:characterize-interval:coeffs} of \Cref{alg:characterize-interval} adds all coefficients of all polynomials to the projection.~\cite[Algorithm 6]{abraham2021} proposes an optimization, which adds fewer coefficients and could also be applied here.

We further note that the presented projection is based on McCallum's projection operator~\cite{mccallum1998}. An adaption of the projection operator to Lazard's projection~\cite{lazard1994,mccallum2019} is possible, as discussed in~\cite[Section 4.4.6]{abraham2021}; this requires an adaption of the added coefficients and root isolation, and changes the correctness arguments (e.g. the implicit cells do not maintain order-invariance, but valuation-invariance).

\subsection{Example}

As significant portions of the algorithm are taken from the cylindrical algebraic covering method, we again refer to~\cite{abraham2021} for more intuition of unsatisfiable coverings.
In this example, we illustrate how both satisfiable and unsatisfiable cells are characterized for an existentially quantified variable and how coverings of satisfying cells are computed for a universally quantified variable. We consider the following formula with constraints $c_1$, $c_2$ and $c_3$ that are depicted in \Cref{fig:example:constraints}:
\[
	\forall x_1.\; \exists x_2.\; c_1: x_2>3.5 -2(x_1-4)^2 \wedge c_2: (x_1-2)^2+(x_2-2)^2-1>0 \wedge c_3: x_2<3 + 0.25(x_1-4)^2
\]

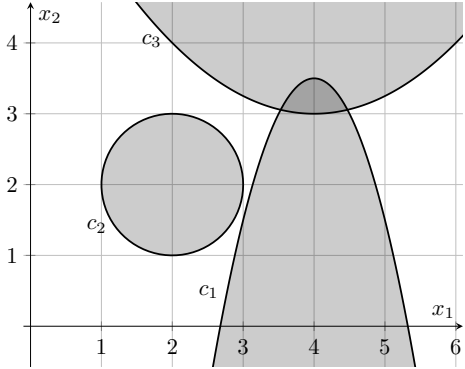
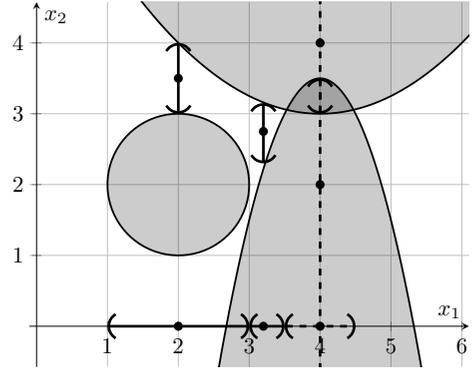
\begin{figure}[t]
	\centering
	\begin{subfigure}{0.45\textwidth}
		\centering
		\begin{tikzpicture}[scale=0.85]
			\begin{axis}[
				axis equal,
				xmin = -0.1,
				xmax = 6.1,
				ymin = -0.1,
				ymax = 4.1,
				xlabel = {$x_1$},
				ylabel = {$x_2$},
				grid = major,
				axis x line = middle,
				axis y line = middle,
			]
				\addplot[name path=c1, thick, no marks, samples=100, fill=black, fill opacity=0.2, domain=2:6] { (3.5 + -2*(x-4)^2) } node[left,pos=0.3,black,opacity=1] {$c_1$};
	
				\addplot[name path=c2,samples=80,domain=0:360,no marks, thick, black, fill=black, fill opacity=0.2] ({cos(x)+2},{sin(x)+2}) node[left,pos=0.6,black,opacity=1] {$c_2$};

				\addplot[name path=c3, thick, no marks, samples=200, fill=black, fill opacity=0.2, domain=1:7] { (3 + 0.25*(x-4)^2) } node[left,pos=0.2,black,opacity=1] {$c_3$};
	
				\path [name path=above] (\pgfkeysvalueof{/pgfplots/xmin},\pgfkeysvalueof{/pgfplots/ymax}) -- (\pgfkeysvalueof{/pgfplots/xmax},\pgfkeysvalueof{/pgfplots/ymax});
				\path [name path=below] (\pgfkeysvalueof{/pgfplots/xmin},\pgfkeysvalueof{/pgfplots/ymin}) -- (\pgfkeysvalueof{/pgfplots/xmax},\pgfkeysvalueof{/pgfplots/ymin});
				
			\end{axis}
		\end{tikzpicture}
		\caption{Graphs of the constraints. The gray areas depict the conflicting cells of the constraints.\\}
		\label{fig:example:constraints}
	\end{subfigure}
	\hfill
	\begin{subfigure}{0.45\textwidth}
		\centering
		\begin{tikzpicture}[scale=0.85]
			\begin{axis}[
				axis equal,
				xmin = -0.1,
				xmax = 6.1,
				ymin = -0.1,
				ymax = 4.1,
				xlabel = {$x_1$},
				ylabel = {$x_2$},
				grid = major,
				axis x line = middle,
				axis y line = middle,
			]
				\addplot[name path=c1, thick, no marks, samples=100, fill=black, fill opacity=0.2, domain=2:6] { (3.5 + -2*(x-4)^2) };
				\addplot[name path=c2,samples=80,domain=0:360,no marks, thick, black, fill=black, fill opacity=0.2] ({cos(x)+2},{sin(x)+2});
				\addplot[name path=c3, thick, no marks, samples=200, fill=black, fill opacity=0.2, domain=1:7] { (3 + 0.25*(x-4)^2) };
	
				\path [name path=above] (\pgfkeysvalueof{/pgfplots/xmin},\pgfkeysvalueof{/pgfplots/ymax}) -- (\pgfkeysvalueof{/pgfplots/xmax},\pgfkeysvalueof{/pgfplots/ymax});
				\path [name path=below] (\pgfkeysvalueof{/pgfplots/xmin},\pgfkeysvalueof{/pgfplots/ymin}) -- (\pgfkeysvalueof{/pgfplots/xmax},\pgfkeysvalueof{/pgfplots/ymin});

				\draw[very thick,{Parenthesis[scale=1.5]}-{Parenthesis[scale=1.5]}] (2,3) -- (2,4);
				\draw[very thick,{Parenthesis[scale=1.5]}-{Parenthesis[scale=1.5]}] (3.2,2.3) -- (3.2,3.15);
				\draw[very thick,-{Parenthesis[scale=1.5]},dashed] (4,-1) -- (4,3.5);
				\draw[very thick,{Parenthesis[scale=1.5]}-,dashed] (4,3) -- (4,5);
						
				\draw[very thick,{Parenthesis[scale=1.5]}-{Parenthesis[scale=1.5]}] (1,0) -- (3,0);
				\draw[very thick,{Parenthesis[scale=1.5]}-{Parenthesis[scale=1.5]}] (3,0) -- (3.5,0);
				\draw[very thick,{Parenthesis[scale=1.5]}-{Parenthesis[scale=1.5]},dashed] (3.5,0) -- (4.5,0);

				\fill (2,0)  circle[radius=2pt];
				\fill (3.2,0)  circle[radius=2pt];
				\fill (4,0)  circle[radius=2pt];

				\fill (2,3.5)  circle[radius=2pt];
				\fill (3.2,2.75)  circle[radius=2pt];
				\fill (4,2)  circle[radius=2pt];
				\fill (4,4)  circle[radius=2pt];

			\end{axis}
		\end{tikzpicture}
		\caption{The satisfiable intervals are indicated with a solid line, the unsatisfiable intervals with a dashed line.}
		\label{fig:example:intervals}
	\end{subfigure}
	\caption{Illustration of the example.}
	\label{fig:example}
\end{figure}

\noindent We start with the first variable being universally quantified:
\begin{description}%
	\item[\texttt{\hyperref[alg:main_forall]{forall}($s = ()$)}] We start covering the real line with satisfiable intervals by sampling values for $x_1$ (\Cref{alg:main_forall:init,alg:main_forall:while} of \Cref{alg:main_forall}). We then sample any value outside the excluded intervals (in this case, we can pick any value); for illustrational purposes (as for all samples in this example), we choose $2$ (\Cref{alg:main_forall:sample} of \Cref{alg:main_forall}). As $\Matrix$ does not evaluate to a value yet, we call the algorithm with the current partial sample to handle the next variable (\Cref{alg:main_forall:call} of \Cref{alg:main_forall}).
	\begin{description}
		\item[\texttt{\hyperref[alg:main_exists]{exists}($s = (2)$)}] We start covering the real line with unsatisfiable intervals (\Cref{alg:main_exists:init,alg:main_exists:while} of \Cref{alg:main_exists}). We sample $x_2 = 3.5$ (\Cref{alg:main_exists:sample} of \Cref{alg:main_exists}) and find a satisfying sample. Now, we generalize to the feasible interval around $(2,3.5)$ as depicted in \Cref{fig:example:intervals}, which is bounded from below by $c_2$ and from above by $c_3$ (\Cref{alg:main_exists:sat} of \Cref{alg:main_exists}). Its projection is the satisfiable interval $(1,3)$ for $x_1$ that we return (\Cref{alg:main_exists:characterize} of \Cref{alg:main_exists}).
	\end{description}
	We store the received satisfying interval (\Cref{alg:main_forall:collect} of \Cref{alg:main_forall}). As there exist samples outside the set of satisfying intervals (\Cref{alg:main_forall:while} of \Cref{alg:main_forall}), we pick the next value $3.2$ for $x_1$ (\Cref{alg:main_forall:sample} of \Cref{alg:main_forall}):
	\begin{description}
		\item[\texttt{\hyperref[alg:main_exists]{exists}($s = (3.2)$)}] We sample $x_2 = 2.75$ and find a satisfying sample. We generalize to the feasible interval bounded by $c_1$ and $c_3$. 
		Note that in the projection of the feasible interval, we take all constraints into account (as all constraints are part of the implicant), even if they do not have a real root at $x=3.2$ -- here, the discriminant of $c_2$ is added to the projection ensuring that no root of $c_2$ crosses the feasible interval. The resulting projection is the satisfiable interval $(3,\underline{3.5})$ for $x$. (The underlined value is an approximation).
	\end{description}
	Similarly, the received interval is stored, and we proceed with the sample $4$ for $x_1$:
	\begin{description}
		\item[\texttt{\hyperref[alg:main_exists]{exists}($s = (4)$)}]
		We sample $x_2 = 4$ (\Cref{alg:main_exists:sample} of \Cref{alg:main_exists}) to obtain the unsatisfiable interval $(3, \infty)$ (\Cref{alg:main_exists:unsat} of \Cref{alg:main_exists}) which we store in the set of unsatisfying intervals (\Cref{alg:main_exists:collect} of \Cref{alg:main_exists}). As this set does not cover the whole real line yet (\Cref{alg:main_exists:while} of \Cref{alg:main_exists}), we sample $x_2 = 2$ (\Cref{alg:main_exists:sample} of \Cref{alg:main_exists}) to obtain the unsatisfiable interval $(-\infty, 3.5)$ (\Cref{alg:main_exists:unsat} of \Cref{alg:main_exists}), which is again stored (\Cref{alg:main_exists:collect} of \Cref{alg:main_exists}). 
		The intervals cover the real line for $x_2$ (\Cref{alg:main_exists:while} of \Cref{alg:main_exists}), as depicted dashed in \Cref{fig:example:intervals}.
		We return the unsatisfiable interval $(\underline{3.5},\underline{4.5})$ for $x_1$ which is the projection of the generalization of the covering (\Cref{alg:main_exists:characterize2} of \Cref{alg:main_exists}).
	\end{description}
	As a recursive call returned an unsatisfiable interval, the algorithm terminates here by returning \unsat (\Cref{alg:main_forall:returnunsat} of \Cref{alg:main_forall}).
\end{description}

\section{Quantifier Elimination}
\label{sec:qe}

From now on, we also allow the input (a formula in prenex normal form) to contain parameters.

To extend the method for quantifier elimination, we could follow a NuCAD~\cite{brown2015nucad} like approach:
we could ``guess'' a sample point for all parameters at once, check the satisfiability of the formula using the method above and construct a cell around the sample point.
We would iterate this by guessing sample points outside the already constructed cells until no such sample points exist.
Finally, we would obtain a list of cells which are either satisfying or unsatisfying.

\begin{algorithm}[t]
	\Data{Global prefix \Prefix{k+1} and matrix \Matrix.}
	\Output{A solution formula for the parameters of $\Prefix{k+1}. \Matrix$.}
	
	\uIf{$k=0$}{
		$(f,O) := $ \FRecurse{$()$} \tcp*{\Cref{alg:main}}
		\lIf{$f = \sat$}{\Return{\tru}}
		\lElse{\Return{\fals}}
	}
	\uElse{
		$(\psi,\cellsymb) := $ \FParameter{$()$} \tcp*{\Cref{alg:main_parameter}}
		\Return{$\psi$}
	}
	\caption{\texttt{user\_call\_qe()} \label{alg:user:qe}}
\end{algorithm}

\begin{algorithm}[t]
	\Data{Global prefix \Prefix{k+1} and matrix \Matrix.}
	\Input{Sample point $s=(s_1,\ldots,s_{i-1}) \in \R^{i-1}$.}
	\Output{$(\psi,\cellsymb)$ where $\psi$ characterizes all satisfying cells over $s$ within $s \times \cellsymb.I$.}

	$\I = \emptyset$ \;
	$\psi := \fals$ \;
	\While{$\bigcup_{\cellsymb \in \I } \cellsymb.I \neq \R$}{
		$s_i :=$ \FSampleOutside{$\I$}  \;
		\uIf{$\Matrix[s \times s_i] = \fals$}{
			$(T,O) := (\fals,\FEnclosingInterval{$s \times s_i$})$ \tcp*{\Cref{alg:get-enclosing-interval}}
		}
		\uElseIf{$\Matrix[s \times s_i] = \tru$}{
			$(T,O) := (\tru,\FEnclosingInterval{$s \times s_i$})$ \tcp*{\Cref{alg:get-enclosing-interval}}
		}
		\uElseIf{$i<k$}{
			$(T,O) :=$ \FParameter{$s \times s_i$} \tcp*{recursive call}
		}
		\Else(it holds $k\leq i<n$) {
			$(f,O) :=$ \FRecurse{$s \times s_i$} \tcp*{\Cref{alg:main}, recursive call}
			\lIf{$f  = \sat$}{
				$T := \tru$ 
			} \lElse{
				$T := \fals$ 
			}
		}
		$\I := \I \cup \{O\}$ \;
		$\psi := \psi \vee (  \FIndexedRootFormula{$O$}  \wedge T )$ \label{alg:main_parameter:formula}
	}

    $\cellsymb :=$ \FCharacterizeCovering{$s$, $\I$} \tcp*{\Cref{alg:characterize-covering}}

	\Return{$(\psi,\cellsymb)$}

	\caption{\texttt{parameter($s$)} \label{alg:main_parameter}}
\end{algorithm}

We propose an alternative approach in \Cref{alg:user:qe,alg:main_parameter} which builds upon the cylindrical algebraic covering method in order to obtain a tree-like description of cylindrically arranged cells which the parameter space. This allows for potentially smaller solution formulas and more intuitive reasoning over its structure, as discussed in~\cite{abraham2021a}.

The idea is to consider the parameters first, and treat them similarly to existentially quantified variables with a few differences:
Instead of returning as soon as we find a satisfiable cell, we collect both satisfiable and unsatisfiable cells until the whole real line is covered by them; analogously to \Cref{alg:main_exists,alg:main_forall}, we compute a generalization of this covering, this time it consists of satisfiable and unsatisfiable cells instead of only satisfiable or only unsatisfiable ones. These ``mixed'' coverings are required to ensure that all satisfiable cells of the parameter space are enumerated. Simultaneously, a symbolic description of the satisfiable cells in the parameters is constructed as a formula in \Cref{alg:main_parameter:formula} of \Cref{alg:main_parameter}.

For the latter, we employ the concept of \emph{indexed root expressions}~\cite{brown1999}:

\begin{definition}[Indexed Root Expression]
	Let $p \in \Q[x_1,\ldots,x_{i}]$ and $j \in \mathbb{N}_{>0}$.
	An \emph{indexed root expression} is a function $\iroot{p}{j}: \R^{i-1} \to \R \cup \{ \mathrm{undefined} \}$; for all $r \in \R^{i-1}$, $\iroot{p}{j}(r)$ is the $j$-th real root of the univariate polynomial $p(r,x_{i}) \in \Q[x_{i}]$ (or undefined if this root does not exist).
\end{definition}

We use constraints over indexed root expressions to describe intervals symbolically: 

\begin{definition}[Indexed Root Formula]
	Let $\cellsymb$ be an implicit cell in main variable $x_i$. The set $\Xi_l(\cellsymb) = \{ \iroot{p}{j} \mid p \in \cellsymb.P \wedge \iroot{p}{j}(s_1,\ldots,s_{i-1}) = \cellsymb.I.l \}$ contains all indexed root expressions bounding the $i$-th component of $\cellsymb$ from below, and the set $\Xi_u(\cellsymb) = \{ \iroot{p}{j} \mid p \in \cellsymb.P \wedge \iroot{p}{j}(s_1,\ldots,s_{i-1}) = \cellsymb.I.u \}$ contains all indexed root expressions bounding the $i$-th component of $\cellsymb$ from above.	
	The \emph{indexed root formula of $\cellsymb$} is the formula $\texttt{indexed\_root\_formula($\cellsymb$)} = \bigwedge_{\xi \in \Xi_l(\cellsymb)} \xi < x_i \wedge \bigwedge_{\xi \in \Xi_u(\cellsymb)} x_i < \xi$.
\end{definition}

\begin{example}
	Consider the polynomials $P = \{ x_2+1,x_1^2+x_2^2-2,x_1-1\}$, the sample point $s=(0,0)$, and the implicit cell $\cellsymb = (P,s,(-1,\sqrt{2}))$. Then the indexed root formula of $\cellsymb$ is $\iroot{x_2+1}{1} < x_2 \wedge x_2 < \iroot{x_1^2+x_2^2-2}{2}$.
\end{example}

\subsection{Simplification of Results}

The above construction of the formula describing the resulting parameter space is naive, as there are low-hanging fruits to simplify these results further. To do so, we store the parameter space in a tree-like \emph{covering data structure} $(\psi_\text{interval}, L)$ where $\psi_\text{interval}$ is an indexed root formula describing some interval symbolically and the \emph{label} $L$ is equal to $\tru$, $\fals$ (we call those \emph{leafs}) or is a list of covering data structures (we call those \emph{inner nodes}) sorted as in \texttt{compute\_cover}. We do not only store the satisfying cells, but also the unsatisfying ones (indicated using the first two values for $L$).

We then simplify by traversing the data structure depth-first: 
\begin{enumerate}
	\item We merge neighboring children $T_1,T_2$ (where $T_1$ is sorted before $T_2$) if both of them do not have children and are labelled both $\tru$ or both $\fals$. The resulting node $T$ is labelled with the common label, and the interval formula is built from the lower bounds of $T_1$ and the upper bounds of $T_2$.  We iterate until we cannot merge any more.
	\item If a node does have a single child which does not have children itself and is labelled either $\tru$ or $\fals$, we erase this child, and take over its label to the current node.
\end{enumerate}
Note that by the merging described in the first step, we remove unnecessary atoms from the output formula: The two corresponding cells are neighbors in the same cylinder and thus overlap in the whole cylinder, allowing us to remove the overlapping bounds. In the second step, the single child corresponds to a cell without bounds, thus the value of the corresponding variable does not affect the output of the algorithm.

The construction of a formula from this data structure is straight-forward by iterating through the tree and encoding the satisfiable cells. However, we propose another optimization during this process, based on the observation that for some subtrees, there are fewer unsatisfying cells than satisfying ones. In those cases, we want to encode the unsatisfying cells instead of the satisfying ones. We encode a leaf in the obvious way, and each inner node as follows:
\begin{enumerate}
	\item We count the number of leaf children labelled with $\tru$ and $\fals$ respectively.
	\item If the first number is smaller, we encode all leaf children labelled with $\tru$ and recursively encode all inner children.
	\item If the latter number is smaller, we encode all leaf children labelled with $\fals$ and recursively encode all inner children. For the latter, we pass a flag indicating that the encoding result should describe unsatisfying cells.
	\item We build the disjunction of all cells. Depending on which of the two cases we entered and the received flag, we negate this disjunction. Afterwards, we build the conjunction of the (possibly negated) disjunction with the indexed root formula of the current node.
\end{enumerate}

Last but not least, we remark that the techniques for CAD-based quantifier elimination could also be adapted for further simplifying the output; in particular the \emph{SimpleSF} algorithm described in~\cite[Section 5.2]{brown1999} is promising. However, to apply it, we would possibly need to refine the given covering to a decomposition first, which would require computational and technical effort.

\subsection{Elimination of Indexed Root Expressions}
\label{sec:indexedroots}

While indexed root expressions are outside the language admitted by standard non-linear real arithmetic, equivalent ``pure'' non-linear real arithmetic formulas can be constructed with some effort. In the following, we discuss some possibilities for doing so, based on common techniques also used in other algorithms. They are not specific to the CAlC method, but worth noting in this context.

A \emph{sign condition} of a set $P$ of polynomials in $x_1,\ldots,x_i$ assigns a sign (positive, negative or zero) to each polynomial; a point in $\reals^i$ \emph{satisfies} a sign condition if the evaluation of each polynomial at $s$ corresponds to its assigned sign.
\emph{Thom's lemma} (see~\cite[pp. 325--326]{mishra1993}) states that the maximal set satisfying some sign condition of a univariate polynomial and all its derivatives is either empty or connected. In other words, we can describe intervals using constraints whose defining polynomials are derivatives of some polynomial.

Given a formula $\varphi$ (containing indexed root expressions) that describes a single cell (which can be extracted efficiently from the covering data structure) and a sample $s \in \reals^i$ that lies in the cell, we can eliminate all indexed root expressions using a generalization of Thom's lemma to multivariate polynomials obtained by viewing them as univariate polynomials in $x_i$ with polynomial coefficients. We may use this generalization in different ways, resulting in the following three different approaches:
\begin{itemize}
	\item Let $P$ be the set of defining polynomials of $\varphi$, and $P' = \{ \frac{\partial^k}{\partial x_i^k} p(x_1,\ldots,x_i) \mid p \in P, i = \level{p}, k=0,\ldots,\vdeg{x_i}{p} \}$ be the set containing all $p \in P$ and all of their partial derivatives $\frac{\partial^k}{\partial x_i^k} p$ with respect to the respective main variable $x_i$. We compute a set of sign conditions on the polynomials in $P'$ such that the union of their described cells is equal to the cell described by $\varphi$. We can do so by starting from $s$, storing the sign condition that satisfies $s$; and then move to an adjacent cell by changing a single sign in the sign condition. If there is a sample $s'$ that satisfies both $\varphi$ and the adapted sign condition, we store the new sign condition. Otherwise we go back. We iterate this until this yields no more sign conditions (when starting from any sign condition in the set); then, all adjacent cells do not satisfy $\varphi$. Finally, we encode all obtained sign conditions by polynomial constraints.
	\item We compute the \emph{augmented projection} of $P$, that is the closure of the CAD projection and partial derivatives w.r.t. the respective main variable, as described in~\cite[p. 144]{collins1975}. This projection yields a CAD that is \emph{projection definable}~\cite[Definition 3]{brown1999}, i.e., the CAD can be encoded using sign-conditions on the projection factors. The difference to the previous approach is the following: all derivatives are delineable in the respective lower-level cells, meaning that the resulting formula is cylindrically arranged. Thus, the resulting formula might require fewer atoms by encoding the cells in a tree-like manner. %
	\item Using techniques from~\cite{brown1999}, we can optimize the previous approach by reducing the amount of derivatives that are added to the projection set - not all derivatives are necessary for obtaining a projection-definable CAD.
\end{itemize}

The computation of derivatives required in the above approaches might lead to large output formulas and additional expensive CAD computations. It is thus desirable to reduce this effort. This is achieved by encoding each cell separately, but encoding all cells in combination in order to eliminate redundancies of sign conditions on the derivatives (effectively reducing the amount of required derivatives). The work in~\cite{brown1999} proposes efficient algorithms for this task; we thus could feed the output formula into these algorithms (either all at once or incrementally as in~\cite[Section 6]{brown1999}).

Certainly, obtaining a small output formula without indexed root expression requires additional computational cost and implementation effort. Moreover, formulas containing indexed root expressions allow for smaller encodings.

\section{Implicant Calculation}
\label{sec:implicants}

The presented algorithms rely on \texttt{implicant\_polynomials} to compute an implicant that generalizes the reason why a formula simplifies to either \tru or \fals at a given sample point, as defined in \Cref{def:implicant}. We propose different variants for computing such an implicant.
In all of these variants, the choice of the implicant is generally not unique. We always define the set of all implicants that can be computed in the described way. In an implementation, we could either collect all possible implicants and decide afterwards which one to take, or, to reduce the computational effort, compute a single good implicant.

We remark that the algorithms presented in this section are kept simple for pedagogical and experimental purposes. We assume that the formula is in prenex normal form and its matrix $\varphi$ is in negation normal form. 

\subsection{Evaluation Only}

We omit Boolean reasoning and simply evaluate the formula using the given sample point.
Let $s \in \reals^i$ such that $\varphi[s] = \fals$. We compute the set of implicants recursively:
\begin{align*}
	\text{implicants}_s(c) & = \{ \}  \text{ if } c \text{ is a constraint and } c(s) \neq \fals \\
	\text{implicants}_s(c) & = \{ \{ \neg c \} \} \text{ if }  c \text{ is a constraint and } c(s) = \fals \\
	\text{implicants}_s(\psi_1 \wedge \psi_2) & = \text{implicants}_s(\psi_1) \cup \text{implicants}_s(\psi_2) \\
	\text{implicants}_s(\psi_1 \vee \psi_2) & = \text{implicants}_s(\psi_1) \times \text{implicants}_s(\psi_2)
\end{align*}
The case $\varphi[s] = \tru$ is analogous, as we basically switch the cases $\psi_1 \wedge \psi_2$ and $\psi_1 \vee \psi_2$.

We interpret the resulting sets of constraints as conjunction.

\subsection{Boolean Propagation}

The above approach misses important Boolean information:
\begin{example}
	Consider $\varphi_1 = (x<0 \vee y=0) \wedge (x<1 \vee y\neq0)$, which is unsatisfiable over $x=2$, but our partial evaluation does not detect that the resulting formula is logically equivalent to $\fals$: As the set of implicants for $y=0$ and $y\neq0$ is empty, the sets of implicants of all subformulas and the formula are empty as well.
\end{example}
We thus incorporate Boolean reasoning that resembles the Boolean propagation implemented within SAT solvers.
Consider \Cref{alg:implicant:base,alg:implicant:propagate}. Let $s \in \reals^i$. If \texttt{implicants($\varphi$, $(\varphi)$, $s$)} (\texttt{implicants($\varphi$, $(\neg\varphi)$, $s$)}) is non-empty, then $\varphi[s] \equiv \fals$ ($\varphi[s] \equiv \tru$), and the algorithm already computes all possible implicants.
The algorithm takes a formula $\varphi$, a sequence of decisions $D$ (its entries are subformulas of $\varphi$ or their negation) and a sample point $s$ as input. The algorithm maintains for each subformula two sets $T(\psi)$ and $F(\psi)$ of \emph{reasons} (which are sets of formulas maintained continuously through the algorithm) implying the subformula evaluates to \tru or \fals, respectively, given the decisions in $D$ and the sample $s$. For every decision $d \in D$, its set of reasons $T(d)$ contains only the empty set, which represents an ``unconditional'' reason. Every subformula $\psi$ and its negation $\neg \psi$ share the same sets (e.g. $T(\psi)$ is the same as $F(\neg \psi)$). Note that the \texttt{evaluate} subroutine corresponds to the evaluation only approach. The algorithm yields an implicant whenever it finds a conflict, that is, some subformula has both reasons to evaluate to \tru and to \fals, contradicting the assumption that $\varphi$ (respectively, $\neg\varphi$) holds.
We note that this scheme can be extended to arbitrary Boolean operators such as exclusive-or and negation.
\begin{example}
	Consider $\varphi_1$ and the sample $x=2$ from the previous example again. Using the sample, we conclude $F(x<0)=\{\{x<0\}\}$ and $F(x<1)=\{\{x<1\}\}$. Additionally, we decide $T(\varphi_1)=\{\emptyset\}$ and propagate $T(x<0 \vee y=0)=\{\emptyset\}$ and $T(x<1 \vee y\neq0)=\{\emptyset\}$. By propagation via the disjunctions, we conclude that $T(y=0) = F(y\neq0) = \{\{x<1\}\}$ and $T(y\neq0) = F(y=0) = \{\{x<1\}\}$. Thus, $T(y=0) \times F(y=0) = \{\{x<0, x<1\}\}$ is the set of implicants.
\end{example}

\begin{figure}[htb]
\noindent
\begin{minipage}[t]{0.54\textwidth}
	\vspace{0pt} 
	\begin{algorithm}[H]
		\caption{\texttt{implicants(}$\varphi$, $D$, $s$\texttt{)}}
		\label{alg:implicant:base}

		\SetKwFunction{FEvaluate}{evaluate}
		\SetKwFunction{FPropagate}{propagate}

		\lForEach{subformula $\psi$ of $\varphi$}{$T(\psi) := \emptyset$}
		\lForEach{subformula $\psi$ of $\varphi$}{$F(\psi) := \emptyset$}
		\lForEach{$d$ in $D$}{$T(d) := \{ \emptyset \}$}
	
		\ForEach{atom $c$ in $\varphi$}{
			\lIf{$c(s)=\tru$}{$T(c) := \{\{ c \}\}$}
			\lElseIf{$c(s)=\fals$}{$F(c) := \{\{ \neg c \}\}$}
		}
	
		\While{$T$ or $F$ changed}{
			\ForEach{subformula $\psi$ of $\varphi$}{
				\FEvaluate{$\psi$} \;
				\FPropagate{$\psi$} \;
			}
		}
	
		\Return{$\bigcup_{\psi \in \varphi} T(\psi) \times F(\psi)$}
	\end{algorithm}

	\begin{algorithm}[H]
		\caption{\texttt{explore(}$\varphi$, $D$, $s$\texttt{)}}
		\label{alg:implicant:explore}

		\SetKwFunction{FImplicants}{implicants}
		\SetKwFunction{FImplicantsExpl}{explore}

		$I_D := $ \FImplicants{$\varphi$, $D$, $s$} \;
		\lIf{$I_D \neq \emptyset$}{
			\Return{$I_D$}
		}
		\textbf{choose} subformula $\psi$ of $\varphi$ s.t. $T(\psi) = F(\psi) = \emptyset$ at the end of the \texttt{implicants} call \;
		\lIf{no such $\psi$ exists}{
			\Return{$\emptyset$}
		}
		$I_{(D,\psi)} := $ \FImplicantsExpl{$\varphi$, $(D,\psi)$, $s$} \;
		$I_{(D,\neg\psi)} := $ \FImplicantsExpl{$\varphi$, $(D,\neg\psi)$, $s$} \;
		\Return{$I_{(D,\psi)} \times I_{(D,\neg\psi)}$}
	\end{algorithm}
\end{minipage}
\begin{minipage}[t]{0.44\textwidth}
	\vspace{0pt} 
	\begin{algorithm}[H]
		\caption{\texttt{evaluate} and \texttt{propagate}}
		\label{alg:implicant:propagate}

		\SetKwProg{Fn}{def}{\string:}{}
		\SetKwFunction{FEvaluate}{evaluate}
		\SetKwFunction{FPropagate}{propagate}

		\Fn{\FEvaluate{$\psi_1 \wedge \psi_2$}}{
			$T(\psi_1 \wedge \psi_2) \ {\cup}{=}\  T(\psi_1) \times T(\psi_2)$ \;
			$F(\psi_1 \wedge \psi_2) \ {\cup}{=}\  F(\psi_1) \cup F(\psi_2)$ \;
		}

		\Fn{\FEvaluate{$\psi_1 \vee \psi_2$}}{
			$T(\psi_1 \vee \psi_2) \ {\cup}{=}\  T(\psi_1) \cup T(\psi_2)$ \;
			$F(\psi_1 \vee \psi_2) \ {\cup}{=}\  F(\psi_1) \times F(\psi_2)$ \;
		}

		\Fn{\FPropagate{$\psi_1 \wedge \psi_2$}}{
			$T(\psi_1) \ {\cup}{=}\  T(\psi_1 \wedge \psi_2)$ \;
			$T(\psi_2) \ {\cup}{=}\  T(\psi_1 \wedge \psi_2)$ \;
			$F(\psi_1) \ {\cup}{=}\  F(\psi_1 \wedge \psi_2) \times T(\psi_2)$ \;
			$F(\psi_2) \ {\cup}{=}\  F(\psi_1 \wedge \psi_2) \times T(\psi_1)$ \;
		}

		\Fn{\FPropagate{$\psi_1 \vee \psi_2$}}{
			$T(\psi_1) \ {\cup}{=}\  T(\psi_1 \vee \psi_2) \times F(\psi_2)$ \;
			$T(\psi_2) \ {\cup}{=}\  T(\psi_1 \vee \psi_2) \times F(\psi_1)$ \;
			$F(\psi_1) \ {\cup}{=}\  F(\psi_1 \vee \psi_2)$ \;
			$F(\psi_2) \ {\cup}{=}\  F(\psi_1 \vee \psi_2)$ \;
		}
	\end{algorithm}
\end{minipage}
\end{figure}

\subsection{Boolean Exploration}

The previous approach is able to detect some Boolean conflicts by propagation. However, there are formulas which are not satisfiable by the Boolean structure already, which can only be detected by exploration:
\begin{example}
	Consider $\varphi_2 = ((z=0 \vee \varphi_1) \wedge (z \neq 0 \vee \varphi_1))$ and the sample $x=2$. This formula is clearly unsatisfiable. However, although we would evaluate the literals ($F(x<0)=\{\{x<0\}\}$ and $F(x<1)=\{\{x<1\}\}$) and decide $T(\varphi_2)=\{ \emptyset \}$, we would not be able to propagate further. We thus need to check the cases $z=0$ and $z\neq0$.
\end{example}
Consider \Cref{alg:implicant:explore}. Let $s \in \reals^i$. If \texttt{explore($\varphi$, $(\varphi)$, $s$)} (\texttt{explore($\varphi$, $(\neg\varphi)$, $s$)}) is non-empty, then $\varphi[s] \equiv \fals$ ($\varphi[s] \equiv \tru$), and the algorithm again computes all possible implicants.
The algorithm naively assumes formulas to evaluate to \tru and \fals, calls recursively, and if both choices lead to a conflict, then it combines the obtained implicants.
We note that also a partial exploration of the search space is possible.
\begin{example}
	Assume we call $\texttt{explore}$ on $\varphi_2$ from the previous example, $D=(\varphi_2)$ and $s=(2)$ (i.e. $x=2$). The implicants call would return the empty set, we thus choose a subformula, for example $\psi := z=0$ thus recursively call $\texttt{explore}$ on $\varphi_2$ and $D=(\varphi,z=0)$ and on $\varphi_2$ and $D=(\varphi,z\neq0)$. The first call would set the reasons of $z=0$ to $\emptyset$, propagate these reasons and detect an immediate conflict from this decision, thus return the implicant $\{ x<0,x<1 \}$; similarly, the second call would return $\{ x<0,x<1 \}$. Thus the overall call would return $\{ x<0,x<1 \} \times \{ x<0,x<1 \} = \{ x<0,x<1 \}$ as implicant.
\end{example}

\subsection{Inprocessing}
\label{sec:inprocessing}

The implicant could be further simplified using preprocessing techniques, such as Gröbner bases (\cite{huang2016} predicts whether preprocessing using Gröbner bases benefits a CAD computation) or techniques described in~\cite{brown2020}. We do not need bookkeeping of relations of input and output constraints, as we do not need to return infeasible subset as in the CDCL(T) framework.

\section{Exploiting the Quantifier Structure}
\label{sec:splitting}

So far, we assumed the input to be in prenex normal form. In the following, we lift this restriction to better exploit the quantifier structure by switching the quantifier order and solving independent subformulas separately. Potentially, we can rule out parts of the formula as being irrelevant in the current branch. To some degree, this is already facilitated by choosing implicants; however, implicants only consider the Boolean structure and completely ignore the quantifiers! 

\begin{example}
	Consider the formula $\varphi = \forall x \forall y.\; (\psi_1(x,y) \wedge \psi_2(x,y))$, which is logically equivalent to $\forall x \forall y.\; \psi_1(x,y) \wedge \forall x \forall y.\; \psi_2(x,y)$. We can check $\forall x \forall y.\; \psi_1(x,y)$ and $\forall x \forall y.\; \psi_2(x,y)$ separately and combine the results accordingly. For each subformula, we can even choose a different variable ordering.
\end{example}

\subsection{Input Transformation}
\label{sec:inputtrans}

We proceed as follows: Assume that we transformed the input formula to negation normal form, that is we pushed the negations into the formula using the double negation rule and De Morgan's rules such that only the atoms occur negatively in the formula. Afterwards, we push the quantifiers as far as possible into the formula using the following rules:

\begin{align}
	Q x Q y.\; \varphi(x,y) & \equiv  Q y Q x.\; \varphi(x,y) \text{ for } Q \in \{ \exists, \forall \}\tag{Swapping Quantifiers}\label{rule:swapping} \\
	Q x.\; \varphi & \equiv \varphi \text{ for } Q \in \{ \exists, \forall \}\tag{Null Quantification} \\
	\forall x.\; (\varphi(x) \wedge \psi(x)) & \equiv \forall x.\; \varphi(x) \wedge \forall x.\; \psi(x)\tag{Distribution over Conjunction} \\
	\exists x.\; (\varphi(x) \vee \psi(x)) & \equiv \exists x.\; \varphi(x) \vee \exists x.\; \psi(x)\tag{Distribution over Disjunction}\\
	Q x.\; (\varphi(x) \circ \psi)  & \equiv (Q x.\; \varphi(x)) \circ \psi\text{ for } Q \in \{ \exists, \forall \},\ \circ \in \{ \wedge, \vee \}\tag{Prenex Law for \{Con,Dis\}junction}
\end{align}

Note that we can apply the \Cref{rule:swapping} for formulas in prenex normal form as well. In this case, it is rather clear how to make use of this rule - we just choose a typical variable ordering that works well for the CAD. In the general case, the order in which we apply \Cref{rule:swapping} and the other rules might affect the result. Thus, the variable ordering has an influence on how we can distribute the quantifiers.

\begin{example}
	Consider the formula $\varphi = \forall x \forall y.\; (\psi_1(x,y) \wedge \psi_2(y))$. When keeping the variable ordering, we can transform the formula to $\forall x.\; (\forall y.\; \psi_1(x,y) \wedge \forall y.\; \psi_2(y))$, although $\psi_2(y)$ does not depend on $x$! When switching the quantifiers, we can transform the formula to $\forall y \forall x.\; \psi_1(x,y) \wedge \forall y.\; \psi_2(y)$.
\end{example}

We could transform the formula iteratively after each assignment of a variable in the covering algorithm by plugging it in to all constraints and propagating the truth values:

\begin{example}
	Consider the formula $\varphi = \forall x \forall y.\; [(x>0 \wedge \psi_1(x,y) \wedge \psi_2(x,y)) \vee (x\leq 1 \wedge \psi_3(x,y))]$. Due to the disjunction, we cannot pull any quantifier into the formula.

	Now assume we picked $x=2$, thus the second case $\vartheta = (x\leq 1 \wedge \psi_3(x,y))$ of $\varphi$ simplifies to $\fals$. An implicant for $\vartheta$ w.r.t. $2$ is $x>1$. We use this information to rewrite the formula to
	\begin{align*}
		& \forall x \forall y.\; [(x>0 \wedge \psi_1(x,y) \wedge \psi_2(x,y)) \vee x>1] \\
		\equiv & \forall x.\; [(x>0 \wedge \forall y.\; \psi_1(x,y) \wedge \forall y.\; \psi_2(x,y)) \vee x>1]
	\end{align*}
	which allows us to split the problem into multiple ones.
\end{example}

\subsection{Adaption of the Algorithm}

\begin{algorithm}[ht!]
	\Input{NRA formula $\varphi$, sample point $s=(s_1,\ldots,s_{i-1}) \in \R^{i-1}$.}
	\Output{$(\sat, \cellsymb)$ or $(\unsat, \cellsymb)$ where $s \times \cellsymb.I$ can or can not be extended to a model for any $s_i \in \cellsymb.I$. In both cases, $\cellsymb$ describes how $s$ can be generalized.
	}

	\uIf{$\varphi[s] \equiv \fals$}{ \label{alg:main_new:false}
		\Return{$(\unsat,\FEnclosingInterval{$\varphi$, $s$})$} \tcp*{\Cref{alg:get-enclosing-interval}}\label{alg:main_new:gei1}
	}
	\uElseIf{$\varphi[s] \equiv \tru$}{ \label{alg:main_new:true}
		\Return{$(\sat,\FEnclosingInterval{$\varphi$, $s$})$} \tcp*{\Cref{alg:get-enclosing-interval}}\label{alg:main_new:gei2}
	}
	\Else(it holds $i<n$){
		\textbf{transform} $\varphi$ based on $s$ \label{alg:main_new:transform} \;
		\uIf{$\varphi = \exists x_{i+1}.\; \psi$}{ \label{alg:main_new:exists}
			\Return{\FExists{$\psi$, $s$}} \tcp*[f]{\Cref{alg:main_exists_new}}
		}
		\uElseIf{$\varphi = \forall x_{i+1}.\; \psi$}{ \label{alg:main_new:forall}
			\Return{\FForall{$\psi$, $s$}} \tcp*[f]{\Cref{alg:main_forall_new}}
		}
		\uElseIf{$\varphi = \bigvee_{j=1,\ldots,k} \psi_j$}{ \label{alg:main_new:vee}
			$\I_\text{unsat} := \emptyset$ \;
			\textbf{sort} $\psi_1,\ldots,\psi_k$ according to some heuristic \;
			\ForEach{$j=1,\ldots,k$}{
				$(f,O) :=$ \FRecurse{$\psi_j$, $s$} \;
				\lIf{$f = \sat$}{
					\Return{$(\sat,O)$}
				} 
				\lElseIf{$f = \unsat$}{
					$\I_\text{unsat} := \I_\text{unsat} \cup \{O\}$ 
				}
			}
			\Return{$(\unsat, \bigcap_{\cellsymb \in \I_\text{unsat}} \cellsymb)$}
		}
		\uElseIf{$\varphi = \bigwedge_{j=1,\ldots,k} \psi_j$}{ \label{alg:main_new:wedge}
			$\I_\text{sat} := \emptyset$ \;
			\textbf{sort} $\psi_1,\ldots,\psi_k$ according to some heuristic \;
			\ForEach{$j=1,\ldots,k$}{
				$(f,O) :=$ \FRecurse{$\psi_j$, $s$} \;
				\lIf{$f = \unsat$}{
					\Return{$(\unsat,O)$}
				} 
				\lElseIf{$f = \sat$}{
					$\I_\text{sat} := \I_\text{sat} \cup \{O\}$ 
				}
			}
			\Return{$(\sat, \bigcap_{\cellsymb \in \I_\text{sat}} \cellsymb)$}
		}
		
	}

	\caption{\texttt{recurse($\varphi$, $s$)} \label{alg:main_new}}
\end{algorithm}

\begin{figure}[htb]
\noindent
\begin{minipage}[t]{0.499\textwidth}
	\vspace{0pt}  
	\begin{algorithm}[H]
		\Input{NRA formula $\varphi$, sample point $s=(s_1,\ldots,s_{i-1}) \in \R^{i-1}$.}
		\Output{see \Cref{alg:main_new}}
	
		$\I_\text{unsat} := \emptyset$ \;
		\While{$\bigcup_{\cellsymb \in \I_\text{unsat}} \cellsymb.I \neq \R$}{
			$s_i :=$ \FSampleOutside{$\I_\text{unsat}$}  \;
			$(f,O) :=$ \FRecurse{$\varphi$, $s \times s_i$} \;
			\uIf{$f = \sat$}{
				$\cellsymb := \FCharacterizeInterval{$s$, $O$}$ \;
				\Return{$(\sat, \cellsymb)$}
			} 
			\uElseIf{$f = \unsat$ }
			{
				$\I_\text{unsat} := \I_\text{unsat} \cup \{O\}$ \;
			}
		}	
		$\cellsymb :=$ \FCharacterizeCovering{$s$, $\I_\text{unsat}$} \;
		\Return{$(\unsat, \cellsymb)$}
		
		\caption{\texttt{exists($\varphi$, $s$)} \label{alg:main_exists_new}}
	\end{algorithm}
\end{minipage}%
\hfill
\begin{minipage}[t]{0.499\textwidth}
	\vspace{0pt}
	\begin{algorithm}[H]
		\Input{NRA formula $\varphi$, sample point $s=(s_1,\ldots,s_{i-1}) \in \R^{i-1}$.}
		\Output{see \Cref{alg:main_new}}
	
		$\I_\text{sat} := \emptyset$ \;
		\While{$\bigcup_{\cellsymb \in \I_\text{sat}} \cellsymb.I \neq \R$}{
			$s_i :=$ \FSampleOutside{$\I_\text{sat}$}  \;
			$(f,O) :=$ \FRecurse{$\varphi$, $s \times s_i$} \;
			\uIf{$f = \sat$}{
				$\I_\text{sat} := \I_\text{sat} \cup \{O\}$ \;
			} 
			\uElseIf{$f = \unsat$}
			{
				$\cellsymb := \FCharacterizeInterval{$s$, $O$}$ \;
				\Return{$(\unsat, \cellsymb)$} \;
			}
		}
		$\cellsymb := $ \FCharacterizeCovering{$s$, $\I_\text{sat}$} \;
		\Return{$(\sat, \cellsymb)$}
	
		\caption{\texttt{forall($\varphi$, $s$)} \label{alg:main_forall_new}}
	\end{algorithm}
\end{minipage}
\end{figure}

We adapt our algorithm to work on general non-linear real arithmetic formulas involving quantifiers and to explore independent subformulas separately, as depicted in \Cref{alg:main_new,alg:main_exists_new,alg:main_forall_new}.
First note that the formula $\varphi$ is now a parameter of all algorithms called recursively. The \texttt{recurse} algorithm (\Cref{alg:main_new}) now handles the cases where the formula is detected to be either equivalent to $\tru$ (\Cref{alg:main_new:true}) or $\fals$ (\Cref{alg:main_new:false}) after plugging in $s$. We note that \texttt{implicant\_polynomials}, which is indirectly called in \Cref{alg:main_new:gei1,alg:main_new:gei2}, needs to be extended for quantified formulas as input, still computing a quantifier-free implicant. We modify \Cref{def:implicant} as follows:

\begin{definition}[Implicant]
	\label{def:implicant_quantifier}
	Let $s \in \R^i$ be a (partial) sample point and $\varphi$ be a formula with free variables $x_1,\ldots,x_i$. If $\varphi[s] \equiv \tru$, then the quantifier-free formula $\psi$ is an \emph{implicant of $\varphi$ with respect to $s$} if
	\[ \psi[s] = \tru \land \left( \psi \Rightarrow \varphi \right) \]
	and the constraints of $\psi$ are of level at most $i$ and contained in $\varphi$.
	Otherwise, if $\varphi[s] \equiv \fals$, then the quantifier-free formula $\psi$ is an \emph{implicant of $\varphi$ with respect to $s$} if
	\[ \psi[s] = \tru \land \left( \psi \Rightarrow \neg\varphi \right) \]
	and the constraints of $\psi$ are of level at most $i$ and contained in $\varphi$.
\end{definition}

\begin{example}
	Consider the formula $\varphi = x_1<0 \vee \forall x_2.\; (x_2>0 \rightarrow x_1>2)$ and assume the sample $(1)$, then $\varphi[(1)]=\fals$ and $0\leq x_1 \wedge x_1\leq 2$ is an implicant for $\varphi$ w.r.t. $(1)$.
\end{example}

If neither of the first two cases hold (\Cref{alg:main_new:false,alg:main_new:true}), we transform the formula based on the current sample point in \Cref{alg:main_new:transform} to facilitate splitting as described at the end of \Cref{sec:inputtrans}.
Then, we do a case distinction on the formula's structure: If the formula is a quantified formula (\Cref{alg:main_new:exists,alg:main_new:forall}), we call the algorithm handling the respective quantifier. If the formula is a conjunction or disjunction (\Cref{alg:main_new:vee,alg:main_new:wedge}), we call \Cref{alg:main_new} recursively on the individual subformulas. Analogously to the exists (forall) case, for disjunctions (conjunctions), we return early once one of the recursive calls returns a satisfying (unsatisfying) cell; otherwise, we collect all unsatisfying (satisfying) cells and build their intersection.

Such an \emph{intersection $\cellsymb'$ of implicit cells $\cellsymb_1,\ldots,\cellsymb_k$} where $\cellsymb_1.s=\ldots=\cellsymb_k.s$  is defined such that $\cellsymb'.P = \cup_{j=1,\ldots,k} \cellsymb_j.P$, $\cellsymb.s=\cellsymb_1.s$, $\cellsymb'.I = \cap_{j=1,\ldots,k} \cellsymb_j.I$. Note that, by definition, the intersection is non-empty, as the implicit cells share the same sample point which is also contained in the intersection. Further, the algorithm only applies the intersection on implicit cells with the same sample point.

We emphasize that the splitting mechanism in \Cref{alg:main_new:vee,alg:main_new:wedge} overlaps with the implicant calculation. In an efficient implementation, this would be interleaved with the calculation of implicants, possibly considering the whole Boolean structure of $\varphi$ at once to traverse the search tree ``non-chronologically''. We could compute valid combinations of recursive calls and choose to ``best'' combination according to some metric.  

\Cref{alg:main_exists_new,alg:main_forall_new} depict the new algorithms for handling existentially and universally quantified variables. Compared to \Cref{alg:main_exists,alg:main_forall}, the calls to \texttt{get\_enclosing\_cell} are moved to \texttt{recurse}. Analogously to \Cref{sec:quantified}, we could extend this approach for parameters in quantifier elimination problems; as this is straight-forward, we omit it here.

\section{Proof System}
\label{sec:proofs}

This section replaces the CAD projection algorithms given above by the proof system introduced in~\cite{nalbach2024levelwise}. This proof system changes the view of ``computing projections of polynomials in one variable less'' as in classical CAD formulations to ``computing properties that a lower-level cell $S$ needs to fulfil such that we can describe a cell in the cylinder $S \times \reals$''; these properties ultimately prove that a polynomial is sign-invariant on a cell.

The motivation is twofold: Firstly, the proof system allows for more efficient projections, as its modular formulation can consider many fine-grained optimizations while keeping the algorithmic aspects clean. Secondly, the proof system might be a step towards a proof-producing procedure for non-linear real arithmetic formulas, i.e., generating (formal) proofs that can be mechanically verified:  While the CAlC method would define a high level proof strategy (as motivated in~\cite{abraham2021a,abraham2020} and mentioned in~\cite{kremer2022}), the proof system would provide a fine(r)-grained layer (as discussed in~\cite[Section 8.1]{nalbach2024levelwise}).

For the following section, we assume the reader being familiar with the proof system in~\cite{nalbach2024levelwise}. We briefly recall the most important definitions before presenting an extension for coverings.

\subsection{Preliminaries}

The work in~\cite{nalbach2024levelwise} introduces a proof system for the \emph{single cell construction}~\cite{brown2015} algorithm, which generalizes unsatisfiable sample points to unsatisfiable cells in MCSAT-based solvers~\cite{kremer2020phd}. The input is a set $P \subseteq \rationals[x_1,\ldots,x_n]$ of polynomials and a sample point $s \in \reals^n$, and we aim to find a description of a cell $R \subseteq \reals^n$ such that $s \in R$ and each $p \in P$ is sign-invariant in $R$. The algorithm iteratively computes \emph{symbolic intervals} $\Isymb_n, \ldots, \Isymb_1$ for the variables $x_n,\ldots,x_1$ such that for each $i=1,\ldots,n$, the bounds of $\Isymb_i$ depend on $x_1,\ldots,x_{i-1}$.

\begin{definition}[Symbolic Interval~\cite{nalbach2024levelwise}]
	A \emph{symbolic interval of level $i$} is either a tuple $(\text{section}, b)$ where $b$ is an indexed root expression with domain $\reals^{i-1}$, or $(\text{sector}, l, u)$ where each of $l$ and $u$ is either $-\infty$/$\infty$ respectively or an indexed root expression with domain $\reals^{i-1}$.
	Intervals of the former represent \emph{sections} $\{ (r,b(r)) \mid r \in R \}$ (where $R \subseteq \reals^{i-1}$ such that $b$ is defined), intervals of the latter represent \emph{sectors} $\{ (r,r') \mid r \in R, r' \in (l(r),u(r)) \}$ (where $R \subseteq \reals^{i-1}$ such that $l$ and $u$ are defined).
\end{definition}

To go into more detail, assume we construct the interval $\Isymb_n$. To maintain sign- resp. order-invariance of $P$ in the resulting cell $R$, we first choose $\Isymb_n$ such that its boundaries are defined by roots of $P$ and $s_n$ is either equal to $b(s_{[n-1]})$ or contained in the interval $(l(s_{[n-1]}),u(s_{[n-1]}))$. We now need to compute the lower-level intervals such that the described underlying cell $\proj{R}{[n-1]}$ is small enough such that $\Isymb_n$ describes a sign-invariant interval for $P$ above each point $r \in \proj{R}{[n-1]}$. To do so, we need to ensure that no root $\xi: \reals^{n-1} \to \reals$ of $P$ crosses a boundary of $\Isymb_n$, meaning that $\xi(r) \leq l(r)$ or $u(r) \leq \xi(r)$ for all $r \in \proj{R}{[n-1]}$. We need to maintain a certain \emph{ordering} on the roots:

\begin{definition}[Indexed Root Ordering~\cite{nalbach2024levelwise}]
	An \emph{indexed root ordering of level $i$} is a relation $\preceq$ on a set $\Xi$ of indexed root expressions with domain $\reals^i$ such that its reflexive and transitive closure $\preceq^t$ is a partial order on $\Xi$. We say it \emph{matches} some $s \in \reals^{i-1}$ if all $\Xi$ are defined at $s$ and $\xi \preceq \xi'$ implies $\xi(s) \leq \xi'(s)$ for all $\xi,\xi' \in \Xi$.
\end{definition}

We first determine an ordering $\preceq$ on all the roots such that $\xi \preceq^t l$ or $u \preceq^t \xi$ for all roots $\xi$ of $P$ that we can ``see'' at $s_{[i-1]}$ (where $\preceq^t$ is the transitive and reflexive closure of $\preceq$). We use CAD projection tools to maintain that this ordering is maintained on $\proj{R}{[n-1]}$, i.e. we compute a set of polynomials $P'$ in $n-1$ variables whose sign-invariance guarantees this property. We iteratively apply the described procedure on the lower levels, until intervals for all variables are computed. 

Observe that above, we allow for some flexibility in the choice of the indexed root ordering by exploiting transitivity. Throughout our procedure, there are more such choices possible, as well as many optimizations in the CAD projection theory which are only applicable in certain cases. A proof system keeps the algorithm maintainable while exploiting these cases: We define \emph{properties of level $i$} which are functions $q: \{ R \mid R \subseteq \reals^i \} \to \{ 0,1 \}$ for $i=1,\ldots,n$. Each proof rule has a single property as consequent; its antecedents are ``smaller'' properties (according to some ordering in the properties) and side conditions which ``enable'' the proof rule.

Given $R \subseteq \reals^i$, we define properties $\sample{s}(R)=1$ iff $s \in R$, $\connected{i}(R)=1$ iff $R$ is connected, $\irordering{\preceq,s}(R)=1$ iff $\preceq$ matches $s$ and $R$ maintains the ordering $\preceq$, $\del{p}(R)=1$ iff the polynomial $p$ is analytically delineable on some connected superset of $R$, and $\representation{\Isymb,s}(R)=1$ iff $\Isymb$ is defined at $s$ and the $i$-th dimension of $R$ is described by $\Isymb$. For more details, we refer to~\cite{nalbach2024levelwise}.

\subsection{Proof Rules for Coverings}

The CAlC algorithm shares similarities with the single cell construction algorithm. In fact, the computation of a single symbolic interval corresponds to \Cref{alg:characterize-interval}. In the following, we will extend the proof system to also cover \Cref{alg:characterize-covering}.

We start by defining a property that holds iff a set of symbolic intervals covers the whole real line if we substitute a sample point $s$:

\begin{definition}
    \label{def:prop:covering}

    Let $i \in \posints$, $R \subseteq \reals^i$, $s \in \reals^{i-1}$, and $\texttt{C} = (\Isymb_1, \ldots, \Isymb_k)$ be a sequence of symbolic intervals of level $i$.

    The property $\covering{\texttt{C},s}$ holds on $R$ if and only if for every $R' \subseteq \reals^{i+1}$ with $\proj{R'}{[i]} = R$ there exists a $j \in \{ 1,\ldots,k \}$ such that the property $\representation{\Isymb_j,s}$ holds on some superset of $R'$.
\end{definition}

This property can be proven using the following rule. It assumes that the intervals are ordered by their lower bounds and are not redundant as in \texttt{compute\_covering}. We then use an indexed root ordering which ensures that the bounds of neighboring intervals overlap.

\begin{lemma}
    \label{def:map:covering}

    Let $i \in \posints$, $R \subseteq \reals^i$, $s \in \reals^{i-1}$, $\texttt{C} = (\Isymb_1, \ldots, \Isymb_k)$ be a sequence of symbolic intervals of level $i$, and $\preceq$ be an indexed root ordering of level $i$.

	Assume that $\texttt{C}$ fulfils the following conditions:
	\begin{itemize}
        \item $\Isymb_1.l=-\infty$ and $\Isymb_k.u=\infty$,
        \item $\Isymb_{j}.l(s) < \Isymb_{j+1}.l(s) \vee (\Isymb_{j}.l(s) = \Isymb_{j+1}.l(s) \wedge (\Isymb_{j} \text{ section} \wedge \Isymb_{j+1} \text{ sector}))$ for $j = 1,\ldots,k{-}1$,
        \item $\Isymb_{j}.u(s) < \Isymb_{j+1}.u(s) \vee (\Isymb_{j}.u(s) = \Isymb_{j+1}.u(s) \wedge (\Isymb_{j} \text{ sector} \wedge \Isymb_{j+1} \text{ section}))$ for $j = 1,\ldots,k{-}1$, and
        \item $\Isymb_{j+1}.l(s) < \Isymb_j.u(s) \vee (\Isymb_{j+1}.l(s) = \Isymb_j.u(s) \wedge (\Isymb_{j+1} \text{ section} \vee \Isymb_j \text{ section}))$ for $j = 1,\ldots,k{-}1$.
    \end{itemize}

    Assume that $\preceq$ matches $s$, and for $j = 1,\ldots,k{-}1$ it holds $\Isymb_{j+1}.l \preceq^t \Isymb_j.u$.

    \begin{align*}
        \sample{s}(R),\ \connected{i}(R),\ \irordering{\preceq,s}(R),\ \\ \forall j = 1,\ldots,k{-}1.\; (\del{\Isymb_j.u.p,s}(R) \wedge \del{\Isymb_{j+1}.l.p,s}(R))   && \vdash \covering{\texttt{C},s}
    \end{align*}
\end{lemma}

\begin{proof}[Proof (Sketch)]
	We ensure that all polynomials defining the lower and upper bounds of the symbolic intervals are analytically delineable on a connected set which contains the current sample point, that means that all their root functions are well-defined over that set. Further, we maintain an ordering of these root functions which ensures that the symbolic intervals cover the cylinder over that cell. For the latter, first observe that the bullet points encode the same requirements as the output of \texttt{compute\_covering}, i.e. that the intervals are sorted according to their lower bounds, that they are overlapping and not redundant. To maintain these overlaps over the underlying set, we require that the indexed root ordering fulfils that the lower and upper bound of all neighboring pairs of symbolic intervals remain in that same order.
\end{proof}

\subsection{Adaption of Algorithms}

We adapt the CAlC algorithm as follows: First, instead of representing an implicit cell with a tuple $(P,s,I)$ where $P$ is a set of polynomials that are sign- or order-invariant in the cell, we represent it as a tuple $(Q,s,I)$ where $Q$ is a set of properties that hold in the cell. We adapt \Cref{alg:get-enclosing-interval,alg:characterize-interval,alg:characterize-covering} to work with the proof system, as given in \Cref{alg:get-enclosing-interval-rules,alg:characterize-interval-rules,alg:characterize-covering-rules}.

\Cref{alg:get-enclosing-interval-rules} initializes the set $Q$ with the sign-invariance of the implicant's polynomials and applies some basic rules such as factorization, and computes the interval above the given sample. \Cref{alg:characterize-interval-rules} adds connectedness of the constructed cell to the set of properties, as the proof system does not always produce descriptions of connected sets. It then isolates the real roots, determines a symbolic interval, an indexed root ordering, and a set $E$ for which we refer to~\cite{nalbach2024levelwise} for details. It then applies all proof rules to the point where the interval above the given sample can be determined. \Cref{alg:characterize-covering-rules} computes a sequence of non-redundant, ordered intervals representing a covering, isolates the roots of each implicit cell, and then determines all symbolic intervals, and a single indexed root ordering which both protects each cell individually and ensures that the boundaries of the symbolic intervals overlap. By choosing a single root ordering, we might be able to rule out redundancies in the root orderings and thus obtain a more efficient projection; however, for now, we do not make use of this possibility and compute orderings for each cell separately, and the trivial ordering that maintains the covering afterwards. The algorithm then applies all proof rules to the point where the interval above the given sample can be determined.

\begin{algorithm}[t]
	\Data{Global matrix \Matrix.}
	\Input{Sample point $s \in \R^{i}$ such that $\Matrix[s] \equiv \fals$ or $\Matrix[s] \equiv \tru$.}
	\Output{A satisfiability-invariant implicit cell $\cellsymb$ containing $s$.} 
	$P := \FImplicant{$\Matrix$, $s$}$ \;
	$Q := \{ \sgninv{p} \mid p \in P \}$ \;
	\textbf{apply} proof rules to $Q$ until only properties $\sgninv{p}$ where $p$ is irreducible remain \;
	$\cellsymb := (Q,s,\FComputeInterval{$s$, $\{ p \mid \sgninv{p} \in Q \}$})$ \tcp*{\Cref{alg:compute-interval}}

	\Return{$\cellsymb$} 

	\caption{\texttt{get\_enclosing\_cell($s$)}}\label{alg:get-enclosing-interval-rules}
\end{algorithm}

\begin{algorithm}[t]
	\Input{Sample point $s \in \R^i$ and an implicit cell $\cellsymb=(\cdot,s \times \cdot,\cdot)$ of level $i+1$.}
	\Output{A satisfiability-invariant implicit cell $\cellsymb$ containing $s$.}
	$(Q,s',\cdot) := \cellsymb$ \;
	$Q := Q \cup \{ \connected{i+1} \}$ \;
	\textbf{compute} set $\Xi$ of symbolic roots of $\{ p \mid \sgninv{p} \in Q \}$ \;
	\textbf{choose} representation $(\Isymb,E,\preceq)$ of $\Xi$ w.r.t. $s'$ \;
	\textbf{apply} proof rules to $Q$ considering $(\Isymb,E,\preceq)$ until only properties $\sgninv{p}$ (of level $i$) where $p$ is irreducible remain \;
	$\cellsymb := (Q,s,\FComputeInterval{$s$, $\{ p \mid \sgninv{p} \in Q \}$})$ \tcp*{\Cref{alg:compute-interval}}
	\Return{$\cellsymb$} 
	
	\caption{\texttt{characterize\_cell($s$, $\cellsymb$)} \label{alg:characterize-interval-rules}}
\end{algorithm}

\begin{algorithm}[t]
	\Input{Sample point $s \in \R^i$ and a set $\I$ of implicit cells of level $i+1$ such that for all $\cellsymb \in \I$, $\cellsymb=(\cdot,s \times \cdot,\cdot)$.}
	\Output{A satisfiability-invariant implicit cell $\cellsymb$ containing $s$.}
	$(\cellsymb_1,\ldots,\cellsymb_k) :=$ \FComputeCover{$\I$} \tcp*{\cite[Section 4.4.1]{abraham2021}}

	\ForEach{$j = 1,\ldots,k$}{
		$(Q_j,s'_j,\cdot) := \cellsymb_j$ \;
		\textbf{compute} set $\Xi_j$ of symbolic roots of $\{ p \mid \sgninv{p} \in Q_j \}$ \;
	}
	\textbf{choose} representations $(\Isymb_j,E_j,\preceq)$ of $\Xi_j$ w.r.t. $s'_j$ for $j=1,\ldots,k$ such that $\preceq$ fulfils the requirement of \Cref{def:map:covering}  \;
	$Q := \cup_{j=1,\ldots,k} Q_j \cup \{ \covering{(\Isymb_1,\ldots,\Isymb_k),s} \}$ \;
	\textbf{apply} proof rules to $Q$ considering $\Isymb_j,E_j$ for $j=1,\ldots,k$ and $\preceq$ until only properties $\sgninv{p}$ (of level $i$) where $p$ is irreducible remain \;
	$\cellsymb := (Q,s,\FComputeInterval{$s$, $\{ p \mid \sgninv{p} \in Q \}$})$ \tcp*{\Cref{alg:compute-interval}}
	\Return{\cellsymb} 
	
	\caption{\texttt{characterize\_covering($s$, $\I$)} \label{alg:characterize-covering-rules}}
\end{algorithm}

\section{Experimental Evaluation}

\subsection{Implementation and Heuristics}
\label{sec:heuristics}

Our implementation incorporates all the algorithms in this paper except the elimination of indexed root expressions for quantifier elimination (\Cref{sec:indexedroots}) and the techniques for exploiting the quantifier structure (\Cref{sec:splitting}); we postponed the implementation of the first due to the high effort, and the latter as this would require deeper changes of our data structures.
We use McCallum's projection operator, which is technically incomplete. However, the implementation of our proof system is complete: In case a polynomial is nullified, we add some of its partial derivatives to ensure its order invariance, as suggested in \cite[Section 5.2]{mccallum1985}.

\subsubsection{Sampling}

When assigning a variable in \texttt{sample\_outside}, we choose the value according to the following scheme: If there are no unsatisfying intervals, we take $0$. Otherwise, we chose an integer below all intervals if possible. Otherwise, we choose an integer above all intervals if possible. Otherwise, we choose a sample point between existing intervals; again, we prefer integers or nice rational numbers if possible, as choosing algebraic numbers leads to expensive computations.

\subsubsection{Variable Orderings}

Variable orderings have a huge impact on the computation of a CAD~\cite{brown2007,dolzmann2004,nalbach2019}.
For technical reasons, our implementation supports static variable orderings only, i.e. we determine a fixed variable ordering based on the set of input constraints and do not adapt the ordering during the computation. This ordering determines the order in which the variables are assigned. However, the CAlC method admits to freely choose any variable to be assigned next; exploiting this is part of future work. 

For quantifier-free formulas, the static orderings are:

\begin{description}
    \item[\textsc{Feature based}] This class of variable orderings computes a set of features of variables within the set of input polynomials (such as average degree, sum of degrees, ...) and sorts the variables by their features (we sort by one feature, break ties using a second or third feature). The first such heuristic was suggested by Brown~\cite{brown2004}. We use a recent improvement obtained using machine-learning techniques from~\cite{pickering2024}.
    \item[\textsc{Max univariate}] Assuming that all preceding variables have been substituted, we select the variable next in which the most constraints are univariate.
\end{description}

To adapt these orderings for quantified formulas, we apply each ordering separately for each \emph{quantifier block}. Two variables $x_i$ and $x_j$ are in the same quantifier block if and only if $Q_{i'} = Q_{j'}$ for all $i',j' \in \{ i,\ldots,j \}$ in the prefix $\Prefix{1}$.

\subsubsection{Implicants}

Experience from~\cite{nalbach2019}, where disabling Boolean decisions in our MCSAT implementation performed best, led to the assumption that Boolean reasoning might make unfavorable decisions for non-linear arithmetic problems. Usually, the algebraic part is harder to solve than the Boolean structure (which is often not complex in the corresponding SMT-LIB benchmark set).

To investigate the impact of the ``algebraic'' complexity on the running time, we implement all three variants for computing implicants (\Cref{sec:implicants}) in a (naive) straight-forward way: We compute all possible implicants and choose the best implicant afterwards (see below).

\begin{description}
	\item[\textsc{Evaluation}] Straight-forward implementation.
	\item[\textsc{Propagation}] We implement full propagation. In a preprocessing step, we add clauses like $\neg (p<0) \vee \neg (p>0)$ to facilitate Boolean propagations: Without this clause, if $p<0$ would be assumed to be true, although we would conclude that $p\leq 0$ cannot hold by Boolean reasoning, we would not for $p>0$ causing additional effort in the theory solving.
	\item[\textsc{Exploration}] Note that this is a rather inefficient implementation of a SAT solver (i.e. without clause-learning etc). Future implementations might consider a more efficient algorithm.
\end{description}

Although we transform the input formula to prenex normal form, we do not eliminate Boolean operators such as exclusive-or and the like, but extend the implicant computation to support these operators.

\subsubsection{Inprocessing}

\begin{description}
	\item[\textsc{Gröbner bases}]  The work in \cite{wilson2012} suggest that preconditioning formulas using Gröbner bases speed up CAD computations. We thus use them for inprocessing (\Cref{sec:inprocessing}), which is applied whenever possible.
\end{description}

\subsubsection{Implicant Selection Heuristic}

After we compute a set of implicants using one of the variants described in \Cref{sec:implicants}, we choose the best according to one of the following criteria:

\begin{description}
    \item[\textsc{Size}] We take the implicant with the minimal number of constraints.
    \item[\textsc{Feature based}] We take modified features from~\cite{pickering2024} to choose the best implicant. For every set of constraints, we compute (1) the sum (over all defining polynomials) of the average of the total degrees of the monomials, (2) the average of the total degrees of all monomials., (3) the sum of the total degrees of all monomials. We sort the sets first by (1), breaking ties with (2), breaking ties with (3).
    \item[\textsc{Sum of total degrees (Sotd)}] We take the implicant with the minimal sum of total degrees, i.e. the sum of total degrees of all monomials of all polynomials in the implicant. The work in~\cite{dolzmann2004} suggests that this predicts the size and computation time of a CAD. We break ties using the size of the implicant.
    \item[\textsc{Reverse sotd}] The opposite of \textsc{Sum of total degrees} (for illustrational purposes).
\end{description}

\subsection{Evaluation}
\label{sec:experiments}

We implemented the CAlC algorithm in our SMT solver \texttt{SMT-RAT}~\cite{corzilius2015}. The implementation is complete for all discussed problems, i.e. we fully support checking quantifier-free and quantified formulas as well as quantifier elimination. Note that our implementation does not convert the input matrix to conjunctive normal form, but directly works on the Boolean structure.
For algebraic computations, we rely on \texttt{libpoly}~\cite{jovanovic2017libpoly}; for factorization and Gröbner bases, we use \texttt{CoCoALib}~\cite{abbott}.
The default variant of \texttt{SMT-RAT} version \texttt{24.02} uses CAlC for quantified problems, and our MCSAT implementation for formulas that can be transformed to quantifier-free formulas. For quantifier-free problems, we also apply standard preprocessing techniques. Support for quantifier elimination needs to be enabled using a flag before compiling.
The tool is available at \href{https://github.com/ths-rwth/smtrat/}{{https://github.com/ths-rwth/smtrat/}}. 

Our implementation is modular in order to evaluate the described variants.
We conduct our experiments on Intel\textregistered Xeon\textregistered Platinum 8160 CPUs with 2.1GHz per core. We use SMT-LIB's \emph{QF\_NRA} (quantifier-free) and \emph{NRA} (with quantifiers) benchmark sets~\cite{barrett2017}.
The source code, instructions for reproducing the experiments and our raw results are all available at \href{https://doi.org/10.5281/zenodo.13366085}{\url{https://doi.org/10.5281/zenodo.13366085}}.

\subsection{Evaluation of Variants}

We start by evaluating the variants of our algorithm for checking satisfiability on the \emph{QF\_NRA} benchmark set due to the greater amount of (non-trivial) benchmarks.
We define a default variant which solved the most instances in preliminary experiments: it uses the \textsc{Max univariate} variable ordering, \textsc{Propagation} for Boolean reasoning, and selects implicants based on the \textsc{Sotd} criteria. All other variants use this configuration, but vary one of these criteria. The results are shown in \Cref{fig:variants}.

\begin{table}
    \caption{Evaluation results of variants. Rows: Number of solved satisfiable and unsatisfiable instances (\emph{sat} and \emph{unsat}), their sum (\emph{solved}), instances where time or memory is exceeded (\emph{timeout} and \emph{memout}). The Default column is the variant that uses \textsc{Propagation} for Boolean reasoning, \textsc{Sotd} for the selection heuristic, \textsc{Max univariate} for the variable ordering, and no preprocessing.}
    \label{fig:variants}

    \centering
	
	\begin{tabular}{l|rr|rrr|r|r|r}
		\toprule
		{} & \multicolumn{2}{l}{Boolean reasoning} & \multicolumn{3}{l}{Selection heuristic} & \multicolumn{1}{l}{Var. order.} & \multicolumn{1}{l}{Inproc.} & \  \\
		{} &       \textsc{Expl.} &   \textsc{Eval.} &     \textsc{Rev. sotd} &   \textsc{Size} & \textsc{Feat. b.} &    \textsc{Feat. b.} &       \textsc{G.B.} & Default \\
		\midrule
		sat         &              5169 &  5153 &             5164 &   5169 &          5174 &              5233 &          5161 &    5184 \\
		unsat       &              4547 &  4738 &             4985 &   5025 &          5046 &              4588 &          5047 &    5048 \\
		timeout     &              2090 &  1801 &             1538 &   1491 &          1471 &              1944 &          1484 &    1454 \\
		memout      &               328 &   442 &              447 &    449 &           443 &               369 &           442 &     448 \\
		solved      &              9716 &  9891 &            10149 &  10194 &         10220 &              9821 &         10208 &   10232 \\
		\bottomrule
	\end{tabular}		  
\end{table}

The choice of the variable ordering has a high impact. One variable ordering is better on the satisfiable instances while the other is on unsatisfiable ones. Further, \textsc{Feature Based} solves $122$ instances not solved by \textsc{Max Univariate}; the other way round, it is $533$. Thus, there is potential for future improvements.

The implicant selection heuristic has a moderate impact, as shown by the numbers of \textsc{Sotd} (see Default column) and its reverse variant \textsc{Rev. sotd}. This indicates a certain variety in the set of implicants from which we select one, but on average only $77\%$ of the generated implicants are used for computing cells. The difference of all meaningful variants (\textsc{Size}, \textsc{Feature Based}, \textsc{Sotd}) is not big, i.e., the virtual best of all variants solves $10250$ instances, only $18$ more than the \textsc{Sotd}. It is unclear whether significant improvements to this heuristics are possible.

Regarding the Boolean reasoning, we observe that \textsc{Propagation} is better than \textsc{Evaluation} (it solves $354$ new instances and loses only $13$ instances), however, \textsc{Exploration} performs worse than \textsc{Evaluation} (it gains only $4$ instances while losing $520$). The first is explained by \Cref{fig:boolean_propagation_off}: \textsc{Propagation} requires far fewer implicants than \textsc{Evaluation}, likely because conflicts are detected earlier using Boolean propagation. We would expect a similar effect in \Cref{fig:boolean_propagation_exploration}, but  \textsc{Exploration} very rarely needs fewer implicants than \textsc{Propagation}; instead, profiling reveals that \textsc{Exploration} spends orders of magnitudes more time in the Boolean propagation and exploration than \textsc{Propagation} (\textsc{Exploration} spends $80\%$ of the time for Boolean reasoning on $114$ instances (which it solved, all of them in than $5$ seconds); meanwhile, \textsc{Propagation} only spends $1\%$ of the time for Boolean reasoning on $110$ of these instances). This is likely due to our rather basic implementation (unoptimized data structures, no clause learning, no watched literals, backtracking is always done to the last UIP instead of the first) to generate \emph{all possible implicants} and choose the best one instead of generating a single good one. Given the large differences on unsat instances in particular, further improvements seem possible.

Also, always applying Gröbner bases to the implicant (as done in the \textsc{Gröbner} variant) results in fewer solved instances than without. As mentioned above, we might need a heuristic that decides when to apply Gröbner bases, as suggested in~\cite{huang2016}.

\begin{figure}
    \begin{subfigure}{0.49\textwidth}
		\centering
        \includegraphics[scale=0.7]{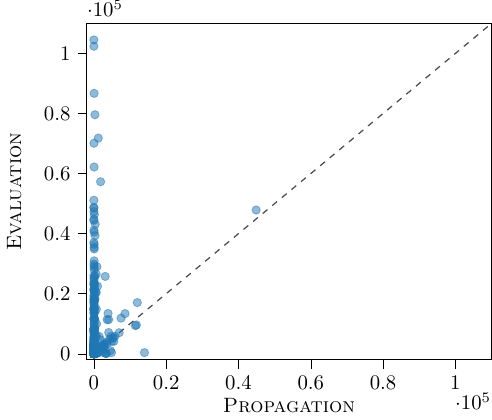}
        \caption{Influence of Boolean propagation.}
        \label{fig:boolean_propagation_off}
    \end{subfigure}
    \begin{subfigure}{0.49\textwidth}
		\centering
        \includegraphics[scale=0.7]{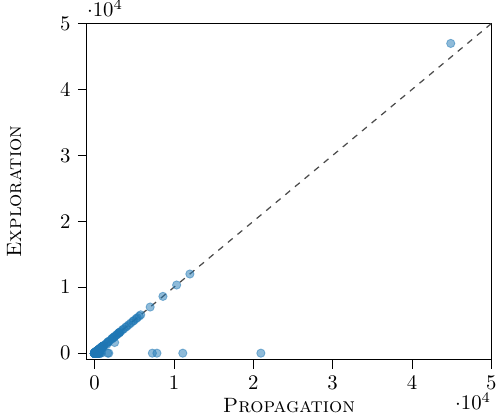}
        \caption{Influence of Boolean exploration.}
        \label{fig:boolean_propagation_exploration}
    \end{subfigure}\hfill
    \caption{Number of implicants that were used for the covering.}
\end{figure}

\subsection{Comparison with Other SMT Solvers}

\begin{figure}
    \begin{subfigure}[t]{0.59\textwidth}
		\centering
        \includegraphics[scale=0.65]{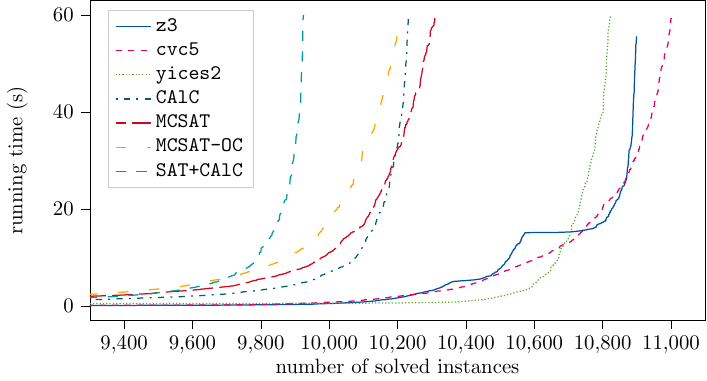}
        \caption{Solved instances on the QF\_NRA benchmarks.}
        \label{fig:results_qfnra}
    \end{subfigure}\hfill
    \begin{subfigure}[t]{0.39\textwidth}
		\centering
        \includegraphics[scale=0.65]{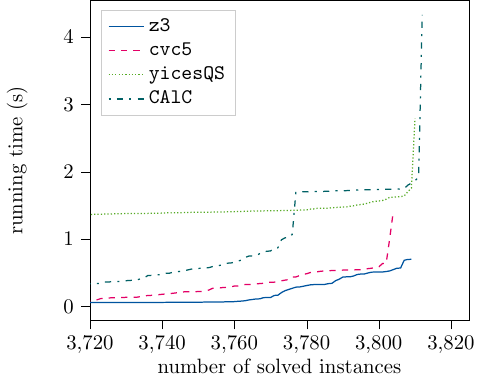}
        \caption{Solved instances on the NRA benchmarks.}
        \label{fig:results_nra}
    \end{subfigure}\hfill
    \caption{Performance profiles of SMT solvers.}
\end{figure}

We now compare the best variant (i.e. the default variant) of our algorithm \texttt{CAlC} with other solvers, both on \emph{QF\_NRA} and \emph{NRA}. For the former, we use \texttt{z3}~4.12.4, \texttt{cvc5}~1.1.0, \texttt{yices2}~2.6.4, an incremental implementation of the CAlC method in our solver as CDCL(T)-backend \texttt{SAT+CAlC} and two MCSAT implementations in our solver, namely \texttt{MCSAT-OC} which uses only the single cell construction from~\cite{nalbach2024levelwise} for theory solving, and \texttt{MCSAT} which additionally uses Fourier-Motzkin, interval constraint propagation~\cite{kremer2020phd}, virtual substitution~\cite{abraham2017}, and subtropical satisfiability~\cite{nalbach2023subtropical}. For the latter, we compare against \texttt{z3}~4.12.4, \texttt{cvc5}~1.1.0 and \texttt{yicesQS}~(Oct 22, 2023).

The results for quantifier-free benchmarks are depicted in \Cref{fig:results_qfnra}.
Clearly, all \texttt{SMT-RAT} variants solve less than the other solvers, partly due to less efficient data structures in particular for large instances. Further, \texttt{cvc5} makes heavy use of linearizations~\cite{kremer2022}, and \texttt{yices2} dynamically changes variable orderings during search.
\texttt{CAlC} is significantly faster than \texttt{SAT+CAlC}. \texttt{CAlC} solves more instances than \texttt{MCSAT-OC} within the given time limit and is generally faster; however, this is expected to change with a higher timeout.

\begin{table}
	\caption{Number of instances, their average number of clauses (after converting them to conjunctive normal form), and their average maximum degree of input polynomials, filtered by instances solved by both or only one (and not the other) solver.}
	\begin{subtable}{0.49\textwidth}
		\caption{Comparison with MCSAT.}
		\label{fig:mcsatcomparison}
		\begin{tabular}{l|ccc}
			\toprule
			& \multicolumn{3}{c}{instances solved by} \\
			& \texttt{MCSAT-OC} & both & \texttt{CAlC} \\
			\midrule
			\# instances & 248 & 9958 & 274 \\
			avg. \# clauses & 1635 & 377 & 320 \\
			avg. max. deg & 2.6 & 5.7 & 8.2 \\
			\bottomrule
		\end{tabular}
	\end{subtable}\hfill
	\begin{subtable}{0.49\textwidth}
		\caption{Comparison with classical CAlC.}
		\label{fig:satcovcomparison}
		\begin{tabular}{l|ccc}
			\toprule
			& \multicolumn{3}{c}{instances solved by} \\
			& \texttt{SAT+CAlC} & both & \texttt{CAlC} \\
			\midrule
			\# instances & 46 & 9879 & 353 \\
			avg. \# clauses & 1975 & 319	 & 1958	 \\
			avg. max. deg & 4.5	 & 5.8	 & 3.8	 \\
			\bottomrule
		\end{tabular}
	\end{subtable}
\end{table}

\Cref{fig:mcsatcomparison} compares \texttt{MCSAT-OC} and \texttt{CAlC}, confirming the impression that \texttt{CAlC} is worse on problems with complex Boolean structure, but has a solid advantage on instances containing hard polynomials. \Cref{fig:satcovcomparison} compares \texttt{SAT+CAlC} and \texttt{CAlC}, yielding a different picture: The instances solved by both solvers have relatively simple Boolean structures, but the instances solved by only one solver have complex Boolean structures - suggesting that the CAlC method particularly requires ``luck'' on problems with Boolean structure.

The results for quantified benchmarks in \Cref{fig:results_nra} look very promising for \texttt{CAlC}, as it solves $6$ instances not solved by \texttt{yicesQS} while loosing $4$. However, we should not draw further conclusions because most instances are solved quickly by all solvers and only $3$ instances remain unsolved.

We further note that all solvers agreed on the same status for each instance.

\subsection{Comparison with Other Quantifier Elimination Tools}

Finally, we evaluate \texttt{CAlC} against \texttt{QEPCAD~B} (used as backend through \texttt{Tarski}~1.28~\cite{vale-enriquez2018}) and \texttt{Redlog}~svn6658. We do not evaluate against commercial tools such as \texttt{Maple} or \texttt{Mathematica} because of the necessary licences.
We evaluate the tools on two benchmark sets: Firstly, we use a collection of CAD examples by David Wilson from Bath University~\cite{wilson2013}, consisting of $78$ formulas encoding relevant mathematical statements, including quantifier alternations and parameters. As this set is small, we also use SMT-LIB's \emph{NRA} benchmarks; some of these instances contain rational functions and other peculiar features of SMT-LIB, and we only use those instances which can be converted to inputs for \texttt{QEPCAD~B} and \texttt{Redlog} in a straight-forward way. All scripts for converting these benchmarks to the respective input formats are provided in the Zenodo repository.

\begin{figure}
	\begin{subfigure}[t]{.49\textwidth}
		\centering
		\includegraphics[scale=0.7]{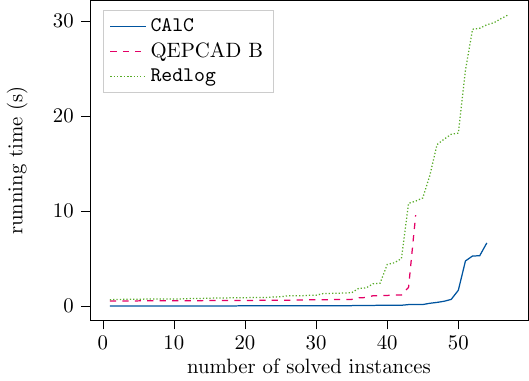}
		\caption{Solved instances on Bath's benchmarks.}
		\label{fig:results_qe}
	\end{subfigure}\hfill
	\begin{subfigure}[t]{.49\textwidth}
		\centering
		\includegraphics[scale=0.7]{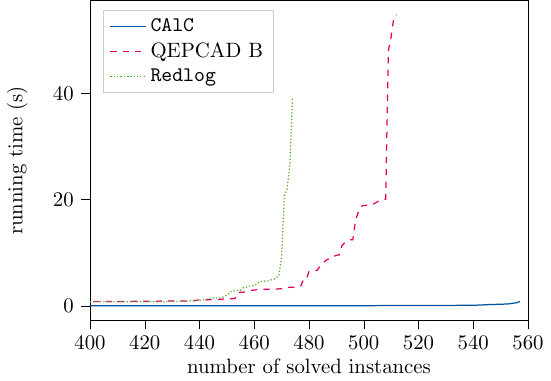}
		\caption{Solved instances on the NRA benchmarks from SMT-LIB.}
		\label{fig:results_qe_nra}
	\end{subfigure}

	\caption{Performance profile  of quantifier elimination tools.}
\end{figure}

The running times are depicted in \Cref{fig:results_qe,fig:results_qe_nra}. We note that \texttt{QEPCAD~B} fails on $25$ instances due to the incompleteness of McCallum's projection. \texttt{CAlC} is competitive on the Bath benchmarks, solving $9$ benchmarks more than \texttt{QEPCAD~B} and $4$ benchmarks less than \texttt{Redlog}. On the SMT-LIB benchmarks, \texttt{CAlC} outperforms the other solvers significantly.

\begin{figure}
	\begin{subfigure}[t]{.49\textwidth}
		\centering
		\includegraphics[scale=0.7]{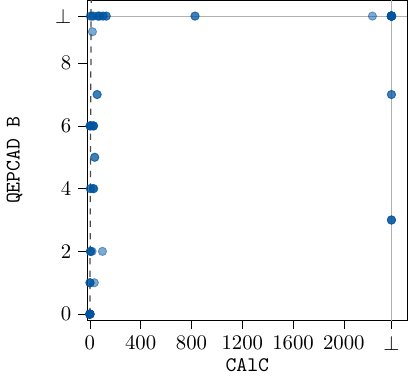}
		\caption{\texttt{CAlC} vs \texttt{QEPCAD~B}.}
		\label{fig:smtrat_qepcad}
	\end{subfigure}\hfill
	\begin{subfigure}[t]{.49\textwidth}
		\centering
		\includegraphics[scale=0.7]{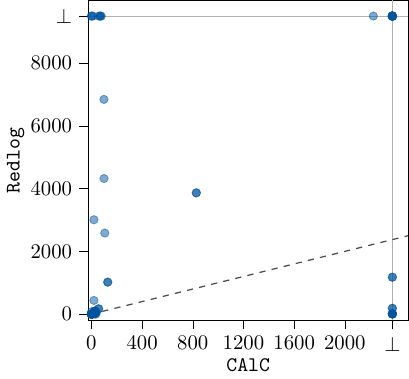}
		\caption{\texttt{CAlC} vs \texttt{Redlog}.}
		\label{fig:smtrat_redlog}
	\end{subfigure}

	\caption{Quality of the solution formula: We compare the number of atoms in the solution formula (indicated by the coordinates). Every point that is not on one of the gray lines represents an instance solved by both solvers. The $\bot$ line indicates that the corresponding solver did time out on the instance.}
\end{figure}

To measure the quality of the generated solution formulas describing the parameter space, we count the number of atoms, as depicted in \Cref{fig:smtrat_qepcad,fig:smtrat_redlog}. Clearly, \texttt{QEPCAD~B} is superior to \texttt{CAlC}. Still, \texttt{CAlC} produces significantly smaller solution formulas than \texttt{Redlog} on many instances.

To verify the correctness of the solution formulas produced by \texttt{CAlC}, we used \texttt{Tarski}~1.40 to check whether they are equivalent to the input problems. To our knowledge, \texttt{Tarski} is the only system that supports indexed root expressions in the input, which may be contained in \texttt{CAlC}'s solution formulas.

\section{Conclusion and Future Work}
\label{sec:conclusion}

\subsection{Future Work}

\paragraph{Preprocessing}

The preprocessing techniques implemented in \texttt{SMT-RAT} are designed for quantifier-free formulas. Incorporating incomplete techniques such as \emph{virtual substitution} for quantified formulas might help to improve the performance on some benchmarks.

\paragraph{Dynamic Variable Orderings}

Variable orderings have a crucial impact on the performance of CAD-based algorithms. The CAlC algorithm allows for different variable orderings for every branch and the algorithm can naturally combine sub-results. This is technically possible in MCSAT as well -- however, combining sub-results computed with different variable orderings comes with high costs when combining projection results that stem from different variable orderings.
The advantages of CAlC could be further facilitated by employing the techniques described in \Cref{sec:splitting} to split the input into more branches, where, again, a different variable ordering is possible in each branch.

\paragraph{Implicant Calculation}

The experimental evaluation shows that the computation of implicants plays a central role in the algorithm. While the choice based on algebraic criteria (e.g. degrees of polynomials) is important, (efficient) Boolean reasoning is crucial. Our implementation features only a basic implementation for Boolean reasoning which exhaustively computed all possible conflicts in order to choose the ``best'' one according to algebraic criteria. Future implementations should focus on efficient Boolean reasoning, incorporating techniques from SAT solvers -- or even using a SAT solver -- guided by algebraic criteria in order to compute a single implicant which is ``good''. 

\paragraph{Reduce Lifting over Sections}

The work in~\cite{bar2023} extends the CAlC algorithm for closed cells, i.e., it tracks whether a truth-invariant cell maintains the same truth value on its closure, based on strict relation symbols in the input formula. This allows to build coverings with closed intervals in the CAlC algorithm. Thereby, we could avoid exploring the branches on some cell boundaries which oftentimes involve computations with non-rational real algebraic numbers that are particularly computationally expensive.
This technique could be extended to CAlC for quantifiers in a straight-forward way.

\paragraph{Minimizing Solution Formulas for Quantifier Elimination}

The CAlC algorithm has a competitive running time for quantifier elimination, however, in particular \texttt{QEPCAD~B} computes significantly smaller solution formulas for the parameters. With some effort, the techniques employed by \texttt{QEPCAD~B} to minimize the solution formula could also be applied for results computed by CAlC.

\paragraph{Optimization Problems}

\emph{Optimization Modulo Theories}~\cite{bigarella2021} deals with the optimization variant of SMT, where we are not only interested in \emph{some} solution that satisfies the input formula, but a solution where the value of a specified \emph{objective variable} is minimal or maximal. A naive approach would be to transform the problem to a quantifier elimination problem where the objective variable is the only parameter. We then compute its solution space and pick the minimal or maximal value. Less naively, we could compute only the part of its solution space that is sufficient to prove that a certain value is the minimal or maximal value.

\subsection{Conclusion}

We generalized the successful CAlC algorithm to quantified input formulas and quantifier elimination problems. Our algorithm works directly on formulas with arbitrary Boolean and quantifier structure by shifting Boolean reasoning to computation of implicants that explain a conflict.
This avoids the need for a complex CDCL(T) architecture and keeps the implementation relatively simple and compact.

Further, we discussed various extensions of our basic algorithm in order to improve the running times and to reduce the solution formula size for quantifier elimination. We presented an embedding into a proof system for cylindrical algebraic decomposition which allows generating certificates to verify the algorithm's results in future implementations.

Our algorithm uses an adapted concept of implicants to incorporate Boolean reasoning. Their computation plays a central role for the performance of the algorithm. For our implementation, we focused on generating optimal implicants with respect to algebraic criteria, mostly neglecting the efficiency of Boolean reasoning. We investigated different ways for computing implicants, varying in the exhaustiveness of Boolean reasoning and the algebraic criteria. Although it was shown experimentally that algebraic criteria have an impact, the computational effort spent on Boolean reasoning limits the scalability of the current implementation. 

Still, our algorithm shows decent performance compared to other tools: On SMT-LIB's QF\_NRA benchmarks, it is comparable with our MCSAT implementation, while state-of-the-art SMT solvers outcompete both implementations. On SMT-LIB's NRA benchmarks, our algorithm solves the most instances, outcompeting the state-of-the-art SMT solvers. Compared to quantifier elimination tools, our algorithm is competitive with respect to running times. However, regarding the output formula size, the state-of-the-art tool \texttt{QEPCAD~B} produces smaller formulas. 

Particularly the latter comparison shows that we successfully transferred ideas from SMT solving to quantifier elimination. The results are promising and motivate future work on the algorithm. Throughout the paper, we presented various possibilities for further improving its performance.

\section*{Acknowledgements}

Jasper Nalbach was supported by the Deutsche Forschungsgemeinschaft (DFG, German Research Foundation) as part of RTG 2236 \emph{UnRAVeL} and AB 461/9-1 \emph{SMT-ART}.

Thanks to Erika Ábrahám for discussions on the CAlC method, and to Christopher W. Brown for the discussion on indexed root expressions and Thom's lemma, as well as extending Tarski to support indexed root epxression in the input formulas.
We thank Philip Kroll for initiating the work on \Cref{sec:splitting} and implementing parts of this work in his thesis.

\printbibliography

\end{document}